\documentclass[aos, preprint]{imsart}

\RequirePackage{amsthm,amsmath,amsfonts,amssymb}
\RequirePackage[numbers]{natbib}
\RequirePackage[colorlinks,citecolor=blue,urlcolor=blue]{hyperref}
\RequirePackage{graphicx}
\usepackage{array}
\usepackage{algorithm,algpseudocode}

\theoremstyle{thmstyleone}%
\theoremstyle{thmstyletwo}%
\newtheorem{remark}{Remark}%
\theoremstyle{thmstylethree}%
\newtheorem{definition}{Definition}

\usepackage{enumitem}
\usepackage{mathabx}
\usepackage{bbm}
\usepackage{bm}
\usepackage{pgf,tikz}
\usetikzlibrary{trees}
\usepackage{booktabs}
\usepackage{multirow}
\usepackage{arydshln}

\newtheorem{prop}{Proposition}[section]
\newtheorem{lem}[prop]{Lemma}

\newtheorem{cor}[prop]{Corollary}
\newtheorem{theo}[prop]{Theorem}

\newcommand{\R}{\mathbb{R}}
\newcommand{\Z}{\mathbb{Z}}
\newcommand{\N}{\mathbb{N}}

\renewcommand\P{\mathbb{P}}
\newcommand{\cara}[1]{\mathbbm{1}_{#1}}
\newcommand{\der}{\mathrm{d}}

\newcommand{\E}{\mathbb{E}}
\newcommand{\V}{\mathbb{V}}
\renewcommand{\epsilon}{\varepsilon}
\newcommand{\1}{\mathbbm{1}}
\newcommand{\diam}{{\rm diam}}

\newcounter{exno}
\newcommand{\ex}[1]{\refstepcounter{exno}\label{#1}}

\begin{document}

\begin{frontmatter}
		\title{Nonparametric intensity estimation of spatial point processes by random forests}

		\runtitle{Random forest intensity estimation}
		
		\begin{aug}
			\author[A]{\fnms{Christophe A.N.}~\snm{Biscio}\ead[label=e1]{christophe@math.aau.dk}}
			\and
			\author[B]{\fnms{Fr\'ed\'eric}~\snm{Lavancier}\ead[label=e2]{frederic.lavancier@ensai.fr}}
			\address[A]{Department of Mathematical Sciences, Aalborg University \printead[presep={,\ }]{e1}}
			\address[B]{Univ Rennes, Ensai, CNRS, CREST\printead[presep={,\ }]{e2}}

		\end{aug}
		
		\begin{abstract}
			We propose a random forest estimator for the intensity of spatial point processes, applicable with or without covariates. 
It retains the well-known advantages of a random forest approach, including the ability to handle a large number of covariates, out-of-bag cross-validation, and variable importance assessment. Importantly, even in the absence of covariates, it requires no border correction and adapts naturally to irregularly shaped domains and  manifolds. Consistency and convergence rates are established under various asymptotic regimes, revealing the benefit of using covariates when available. Numerical experiments illustrate the methodology and demonstrate that it performs competitively with state-of-the-art methods.
		\end{abstract}
		
		\begin{keyword}[class=MSC]
			\kwd[Primary ]{62M30}
			\kwd{62G05}
			\kwd[; secondary ]{60G55}
		\end{keyword}
		
		\begin{keyword}
			 \kwd{Inhomogeneous spatial point process}
			\kwd{Intensity function}
			\kwd{Nonparametric intensity estimation}
			\kwd{Random forest}
			\kwd{Variable importance}
		\end{keyword}
		
	\end{frontmatter}

\section{Introduction}

Spatial point patterns are ubiquitous in many fields, including biology, ecology, epidemiology, criminology, astronomy, and materials science. The first crucial step in the statistical analysis of such data is to estimate the intensity of points over space, which provides information on the average number of points per unit area.

Depending on the data at hand, two main strategies can be considered to estimate the intensity. The first accounts for spatial variation solely through the spatial coordinates of the observed points. The second  leverages spatial covariates, if available, that may explain the spatial fluctuations of the intensity.  
A brief review of the existing literature in these two settings is discussed  in Section~\ref{sec:literature}. 
In this contribution, we introduce a random forest approach for  nonparametric intensity estimation, applicable with or without covariates. 
What we call (with slight abuse) a tree in this random forest is a random partition of the spatial domain or, if covariates are available, of their image space, where the intensity in each cell is simply estimated as the number of observed points divided by the cell volume. As usual, the random forest estimator is obtained by averaging over different trees.
The construction is formalised below in Section~\ref{sec:methodology}.

\smallskip

Random forests are popular for regression and classification due to their flexibility,  accuracy and ability to handle a large number of covariates. Since Breiman's seminal paper  \cite{Breiman01}, the literature on random forests has grown considerably. We refer to the textbook \cite{hastie2009elements} for basic concepts, to   \cite{biau-scornet-15, Scornet26} for theoretical reviews and to \cite{Boulesteix12,Probst19} for methodological aspects. Our adaptation to the point process setting preserves the aforementioned advantages while offering the following benefits:
\begin{itemize}
\item Each tree estimator of the intensity is easily interpretable, including through visual inspection;
\item The method naturally accommodates the geometry of the observation domain, without requiring any additional edge correction, making it particularly well suited to irregularly shaped domains and domains with occlusions; 
\item Similarly, the estimator can be straightforwardly applied to point patterns observed on a manifold,  without requiring shape corrections; 
\item The method comes with the standard analytical tools available for random forests.  In particular, hyperparameter selection can be easily performed using out-of-bag cross-validation;
\item For the same reason, when covariates are involved, their importance in estimating the intensity  can be quantified.
\end{itemize}

Moreover, our method is supported with strong theoretical guarantees. We establish the consistency of the estimator and its rate of convergence under different asymptotic regimes: infill asymptotics, increasing domain asymptotics, or a combination of the two. Infill asymptotics is appropriate  when observed points are dense within a fixed domain. Increasing domain asymptotics, on the other hand, applies when points are not necessarily dense, but numerous because the observation domain is very large. 
It is well known that the consistency of intensity estimation in these regimes depends on whether the estimator leverages covariates \citep{guan2008,lieshout21Infill}. We confirm this conclusion for our random forest estimator. Moreover, we go a step further by examining, via the rate of convergence, the benefit of leveraging covariates when available, compared to an estimation based solely  on spatial coordinates. We show that, overall, leveraging covariates is generally beneficial, provided the intensity genuinely depends on them. 
Our theoretical findings can be summarized as follows:
\begin{itemize}
\item In an increasing domain asymptotic regime, consistency is achieved only if covariates are leveraged, provided they take the same values sufficiently often across the observation domain. 

\item In an infill asymptotic regime, leveraging covariates or not generally leads to consistent estimation, which can even achieve the minimax rate of convergence when the underlying intensity is assumed to be H\"older continuous. However, in certain situations where covariates are locally very smooth, leveraging them can lead to a strictly faster rate of convergence.

\item In an intermediate asymptotic regime, consistency generally holds whether covariates are leveraged or not, but the rate of convergence is typically faster if covariates are used, provided they are sufficiently smooth locally and take repeated values frequently enough in space. 
\end{itemize}

To conduct this theoretical study, we consider a purely random forest estimator,  where the tessellations for each tree are constructed independently of the observed point pattern. This setting coincides exactly with our method when the intensity estimator does not  involve covariates, but it constitutes a simplifying assumption when covariates are leveraged. Investigating the theoretical properties of  genuine random forests is notoriously difficult \citep{Biau08,Scornet15,biau-scornet-15}, and this simplification allows for deeper theoretical insights  \citep{biau12,arlot2014,mourtada2020, klusowski21b,oreilly2021}. From a broader perspective,  genuine random forests are believed  to generally outperform  purely random forests (see \cite{Mourtada21} for numerical illustrations), so that  theoretical results established for purely random forests can be interpreted as worst-case guarantees.

\subsection{Random forest intensity estimator} 
\label{sec:methodology}

Let $X$ be a spatial point process on $\R^d$, $d\geq 1$. Assuming its existence, the intensity of $X$ is the function $\lambda$ satisfying, for any Borel set $B \subset \R^d$,
\begin{equation*}
\E |X\cap B| = \int_{B} \lambda(u) \der u,
\end{equation*}
where $|X\cap B|=  \sum_{u\in X} \cara{u \in B}$ denotes the number of events in $B$. 
Our aim is to estimate  $\lambda$ based on a single realisation of $X$ observed on a  bounded set $W\subset \R^d$. We assume throughout that $W$ is known. Extending our approach to the case where $W$ is inferred from the observed point pattern is beyond the scope of this article.

We moreover assume that a $p$-dimensional covariate $z:\R^d \to \R^p$, $p\geq 1$, may be observed on $W$ and that the intensity $\lambda$ depends on $z$, that is, for all $u\in\R^d$,  $$\lambda(u)=f(z(u))$$ for some nonnegative function $f$. 
Note that the particular case $z(u)=u$, for all $u\in\R^d$, reduces to the situation where no covariate is available, and the intensity  simply depends on the spatial coordinates. This specific situation will be discussed in Section~\ref{sec:purelyspatialpart}.

Based on a realisation on $W$, we shall estimate $\lambda(x)$ for any $x$ such that $z(x)\in z(W)$, which of course includes any $x\in W$, but potentially many more. Our estimator is based on partitions of  $z(W)$. The core of the method lies in how these partitions are built. We detail this construction in the two following sections, depending on whether covariates are available or not. Given these partitions, the random forest intensity estimator is constructed as follows. 

Let $\{I_{j}, j\in \mathcal J\}$ be a
 partition of $z(W)$, so that
$$z(W)=\bigcup_{j \in \mathcal J} I_{j},$$
and the $I_j$'s do not overlap, meaning that they have no common interior points. 
For $x$ such that $z(x)\in z(W)$, we denote by $I(x)$ the set $I_{j}$ such that $z(x)\in I_{j}$. We assume that this set is unique, which means that the partition is such that $z(x)$ does not belong to a boundary $I_{j_1}\cap I_{j_2}$ for $j_1\neq j_2$. 
We further let $A_{j}=z^{-1}(I_{j})\cap W$ be the inverse image of
$I_{j}$ in $W$ and we denote $$A(x)=z^{-1}(I(x))\cap W.$$ Note that
$\{A_{j}, j\in \mathcal J\}$ forms a partition of $W$.

We call  {\it tree} intensity estimator of $\lambda(x)$, based on the above
partition, the estimator 
\begin{equation}\label{e:tree estimator without n}
\hat\lambda^{(1)}(x) = \frac{|X\cap A(x)|}{|A(x)|}.
\end{equation}
Here $|X\cap A(x)|$ refers to the number of points of $X$ in $A(x)$, while $|A(x)|$ denotes the volume of $A(x)$. This estimator is piecewise constant on $W$, similar to a histogram  with bins corresponding to the cells $A_j$. Its value in each bin is given by a simple and natural estimate: the number of points in the bin divided by its volume. Consider now a collection of $M$ tree  intensity estimators
$\hat\lambda^{(1)}(x),\dots,\hat\lambda^{(M)}(x)$, each based on
a different partition of $z(W)$. The random forest intensity
estimator of $\lambda(x)$ is 
\begin{equation}\label{e:rf estimator without n}
\hat\lambda^{(RF)}(x) = \frac{1}{M} \sum_{i=1}^M \hat\lambda^{(i)}(x).
\end{equation}

\begin{remark}
Note that $\hat\lambda^{(RF)}$ accommodates the geometry of $W$, including  border effects, since the estimate in each cell $A(x)$ is normalised by its volume $|A(x)|$. In particular we have $\int_W \hat\lambda^{(RF)}(x) dx=|X\cap W|$, which is an unbiased estimator of $\int_W \lambda(x) dx$. This contrasts with standard kernel estimators, for which border corrections are typically required to ensure this property, as the intensity would otherwise be severely underestimated near the boundary of $W$, see for instance \cite{baddeley15book}.  As a result,  no additional border correction is required for our method.
\end{remark}

\subsection{Existing methods for intensity estimation}\label{sec:literature}
In the literature, intensity estimation based solely on spatial coordinates, without including covariates, is most often addressed through kernel smoothing \citep{diggle1985kernel, baddeley15book},
 though alternative nonparametric strategies have also been proposed  \citep{Barr10,baddeley2022diffusion,Flaxman17,sangalli2021complexdomain}. 
From a theoretical point of view, these approaches are consistent under an infill asymptotic regime, but generally not under increasing domain asymptotics \citep{guan2008,lieshout21Infill,sangalli2021complexdomain}. From a practical point of view, a major concern is the treatment of edge effects near the boundary of the observation domain. Various edge corrections have been introduced for kernel smoothing, but these can be challenging to implement in the presence of irregular domains, motivating the development of more sophisticated estimators better suited to such settings 
 \citep{baddeley2022diffusion,sangalli2021complexdomain}. A similar concern arises for point patterns observed on a manifold, where corrections accounting for the underlying geometry may be required \citep{Cohen23}.
Alternatively, Bayesian approaches based on a Gaussian process prior have been considered  \citep{samo15,Donner18,walder17,Aglietti19}, but they face challenging computational issues arising from the need to handle intractable infinite-dimensional distributions.

Nonparametric intensity estimation in the presence of covariates is much less developed. Many studies assume a parametric form, typically of log-linear type \citep{Schoenberg05,waagepetersen:guan:09,choiruddin2023Dantzig}. While convenient for inference, this setting may be overly restrictive in practice.
A natural nonparametric approach consists of kernel smoothing, where the distance between two spatial locations is replaced by the distance between the corresponding covariate values \citep{guan2008,baddeley2012nonparametric}. However, due to the curse of dimensionality, this method is suitable only when the number of covariates is small.
In the spirit of the aforementioned Bayesian approaches, \citep{Kim22} introduced the augmented permanental process (APP) model, where the intensity is assumed, a priori, to be a realisation of a squared Gaussian process indexed by the covariate values. According to the simulation study conducted in \cite{Lu06082025}, see also our Section~\ref{sec:simulation}, this method performs well in the presence of a small number of covariates, but its performance deteriorates considerably as their number increases.
Another line of research involves deep learning-based point process models. However, these models are mostly limited to temporal rather than spatial point processes, relying on recurrent neural networks \cite{Du2016,Mei2017,Xiao2017} or attention-based models \citep{Zuo2020,Zhang2020,zhu2022}.
An exception is the deep mixture point process model of \cite{Okawa19}, later extended to the deep kernel mixture point process (DKMPP) in \cite{zhang2023}. This approach allows for the estimation of spatial and spatio-temporal intensities. However, calibrating the high-dimensional parameter of the model requires a huge number of points for training, with tens of thousands to millions of events considered in these studies. For moderate-sized point patterns, the simulation study of \cite{Lu06082025} shows that its performance is not competitive.

Finally, closely related to our approach in the presence of covariates, an adaptation of the standard XGBoost method to point processes, namely XGBoostPP, is proposed in \cite{Lu06082025}. In this work, the log-intensity is estimated by additive trees, where each tree defines a partition of the covariate space. As in the standard XGBoost algorithm \cite{Chen2016}, a greedy search is used to build each tree by finding the best split among all covariates, while the output value associated with each leaf is optimised numerically. Specific to the point process setting, the loss function is based on the Poisson likelihood, with an additional penalty on tree complexity. In \cite{Lu06082025}, no theoretical guarantees are provided, but the simulation study demonstrates  good empirical performance.
However, the resulting estimator is not straightforward to interpret: the target function is the log-intensity rather than the intensity itself, so that the scores associated with the leaves do not have a direct interpretation. An adaptation of XGBoostPP to the purely spatial case, where no covariates are available, is also lacking. The method moreover requires tuning both a regularisation parameter and a learning rate, and is computationally demanding. We further compare the numerical performance of our method and XGBoostPP in our simulation study.

\subsection{Outline}

The article is organized as follows. 
The procedure for constructing the partitions used in the random forest intensity estimator is described in 
Sections~\ref{sec:purelyspatialpart} and \ref{sec:covariatespart}, depending on whether the estimation relies solely on spatial coordinates, with no covariates observed (i.e., $z(u)=u$ for all $u\in\R^d$), or includes covariates. 
In each case, we present the methodology, discuss the choice of hyperparameters,  and provide illustrations based on real datasets. 
Section~\ref{sec:simulation} contains an intensive simulation study on synthetic data to evaluate our approach and compare it with state-of-the-art methods.  The theoretical analysis is developed in Section~\ref{sec:theory},  with proofs postponed to  Section~\ref{sec:proof}. Finally, the appendices provide additional material on random tessellations, point processes, and our simulation study. The implementation of the estimator and related utilities is provided in the R package \texttt{spforest}, available on our GitHub repository \url{https://github.com/biscio/spforest}. To ensure reproducibility, all experiments presented in this paper can be accessed at  \url{https://github.com/biscio/spforest_simulation_study}.

\section{Purely spatial partitions}
\label{sec:purelyspatialpart}

Adopting the same notation as above, we assume in this section that  $z(u)=u$ for all $u\in\R^d$, which means that no covariate is available and the estimation of $\lambda(x)$  relies solely on the spatial coordinates. In this case, $z(W)=W$ and the cells $I_j$ and $A_j$ coincide.

\subsection{Construction}

 To obtain different partitions of $W$ in the purely spatial setting, we propose to generate independent stationary random tessellations of $W$.
A simple and standard example  is the Poisson Vorono\"i tessellation. 
Let $\{u_i\}_{i\in\N}$ denote the realisation of a homogeneous Poisson point process on $\R^d$, independent of $X$,  with intensity $\gamma>0$, whose choice is discussed below.
The Vorono\"i  cell $V_i$ is the set of all points of $\R^d$ closer to $u_i$ than any other event $u_j$, $i\neq j$. The set of all Vorono\"i cells forms a partition of $\R^d$, i.e., $\R^d=\bigcup_{i\in \N} V_i$, called Poisson Vorono\"i tessellation, see for instance~\cite{chiu2013} for more details. Other standard Poisson-based tessellations, depending on a unique intensity parameter $\gamma$, can be similarly considered, such as the Poisson Delaunay,  Poisson hyperplane and
STIT tessellations, see \cite{schneider2008,chiu2013} and Appendix~\ref{sec:tessellations}. 
Given such a tessellation of $\R^d$ with cells $V_i$,  
we obtain the partition  $W=\bigcup_{i\in \N} V_i\cap W$. 
From a practical point of view, it is enough to generate the tessellation on a rectangular window containing $W$, and consider the intersection with $W$.   However, in some situations, especially if $W$ is a disconnected set, it may happen that for some $i$, $V_i\cap W$ is composed of disjoint subcells. In this case, we consider these subcells as different cells of the partition. 
We finally obtain the partition $W=\bigcup_{j\in \mathcal J} A_j$, where $A_j$ either corresponds to a cell $V_i\cap W$ (if it is a connected non-empty cell) or to a subcell of it.

We can thus construct as many independent tessellations of $W$ as we wish, by simply generating independent realisations of the ancillary homogeneous Poisson point process. The tree intensity estimator \eqref{e:tree estimator without n} is deduced for each of them, leading to  the final random forest estimator \eqref{e:rf estimator without n}. 

Interestingly, this methodology naturally extends  to point patterns on a manifold, as long as  random tessellations can be generated---a straightforward task once the manifold is represented as a fine mesh. This is illustrated in Section~\ref{sec:spatial illustrations}.

\begin{remark}
If the partition is generated by a STIT tessellation, the construction is recursive and the resulting partition forms a genuine tree. However, random tessellations are not necessarily recursive (for instance, Vorono\"i tessellations do not have this property), and the resulting partitions do not always have a tree structure. Nevertheless, with a slight abuse of terminology, we refer to the resulting estimator as a random forest estimator in all cases.
\end{remark}

\subsection{Choice of the hyperparameter $\gamma$}

A crucial parameter in the above construction is the intensity $\gamma$ of the ancillary Poisson point process, from which each tessellation is generated. It  represents the mean number
 of cells per unit measure, so that the number of cells in a region $W$ is on average $\gamma |W|$. This  parameter plays the same role as the bandwidth in kernel estimation, though in an opposite manner: the smaller the value of $\gamma$, the smoother yet more biased the random forest intensity estimator. In our theoretical study in Section~\ref{sec: consistency}, we obtain optimal rates of $\gamma$ for the consistency of the random forest estimator, depending on  the mean number of observed points and on the size of the observation window $W$. 
 In practice, 
 we propose the following rule of thumb, adapted from the Freedman-Diaconis choice of bins' widths for a histogram. Remember that the latter is $\ell=2IQR/\sqrt[3]{N}$ where $IQR$ is the interquartile range of the univariate data at hand, and $N$ is its cardinality. If we view each cell of the tessellation as $d$-dimensional bin with approximated volume $\ell^d$, we need on average $|W|/\ell^d$ cells to cover $W$, which in other words corresponds to an intensity $\gamma=\ell^{-d}$. 
Moreover, since our point pattern is $d$-dimensional, we consider for the interquartile value the mean interquartile range of the point coordinates, denoted by $\widebar{IQR}$. We then obtain as a rule of thumb:
 \begin{equation}\label{choice gamma}
 \gamma=\frac{|X \cap W|^{d/3}}{2^d\widebar{IQR}^d}.\end{equation}
If the point pattern is well spread inside $W$, we can expect that $\widebar{IQR}\approx |W|^{1/d}/2$, as for a uniform distribution, giving the simpler rule: $\gamma=|X\cap W|^{d/3}/|W|$. 
This choice leads to approximately $|X\cap W|^{d/3}$ cells in each tessellation. 
From our experience, it has proven to be a good rule of thumb. 

Alternatively, the choice of $\gamma$ may also rely on an out-of-bag cross-validation procedure, as detailed later in Section~\ref{sec:OOB} for the setting where covariates are used. This requires each tree intensity estimator to be computed from a bootstrap sample of $X$ in $W$. 
In our numerical illustrations, we use the simple rule of thumb described above.


\subsection{Illustrations}\label{sec:spatial illustrations}

Figure~\ref{fig:concrete} shows the result of the intensity estimation of air bubbles in a 5.6 mm$^2$ cross-section of a concrete body. This dataset was studied in \cite{Natesaiyer92, Igarashi} and  is available in the \texttt{R} package \texttt{spatstat} under the name \texttt{concrete}. The centroids of the air bubbles form the point pattern shown at the left panel of the figure. They are located in the cement paste matrix surrounding the grey aggregate particles. In the middle of the figure, a tree intensity estimation based on a random Vorono\"i tessellation is  displayed, while the right-hand plot shows the result of the random forest estimator averaged over 100 random tessellations. The scale of the intensity is the number of points per mm$^2$. Note that the tessellations adapt to the geometry of the region, which is quite irregular in this example, so that no additional border correction is needed for the intensity estimation.

\begin{figure}
\begin{tabular}{m{3.2cm}m{3.8cm}m{3.8cm}} \includegraphics[width=0.29\textwidth]{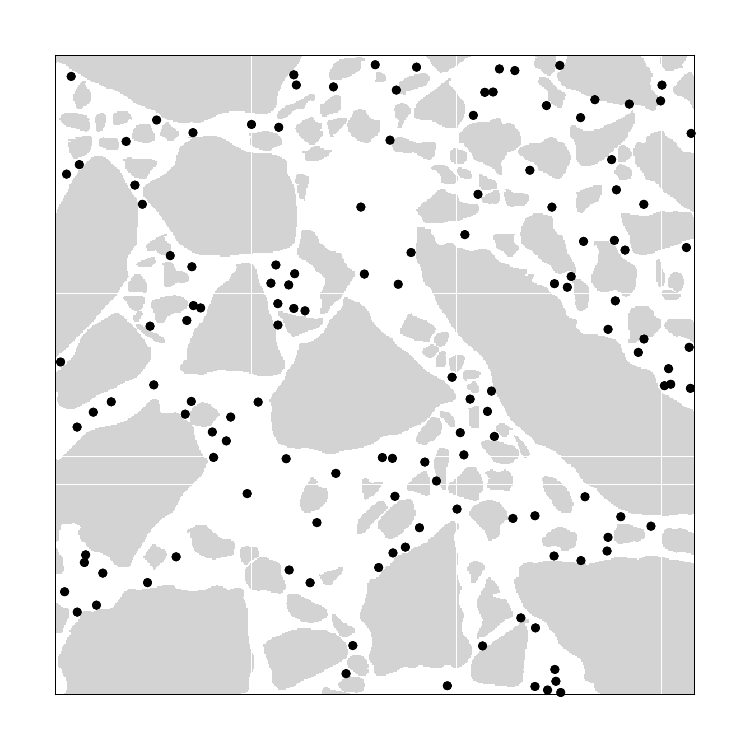} & \includegraphics[width=0.34\textwidth]{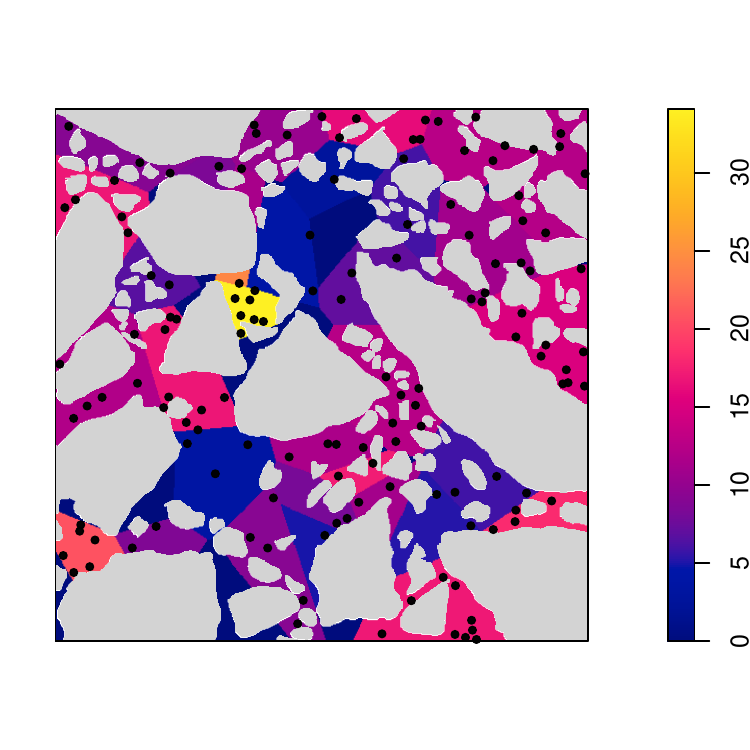}  &
 \includegraphics[width=0.34\textwidth]{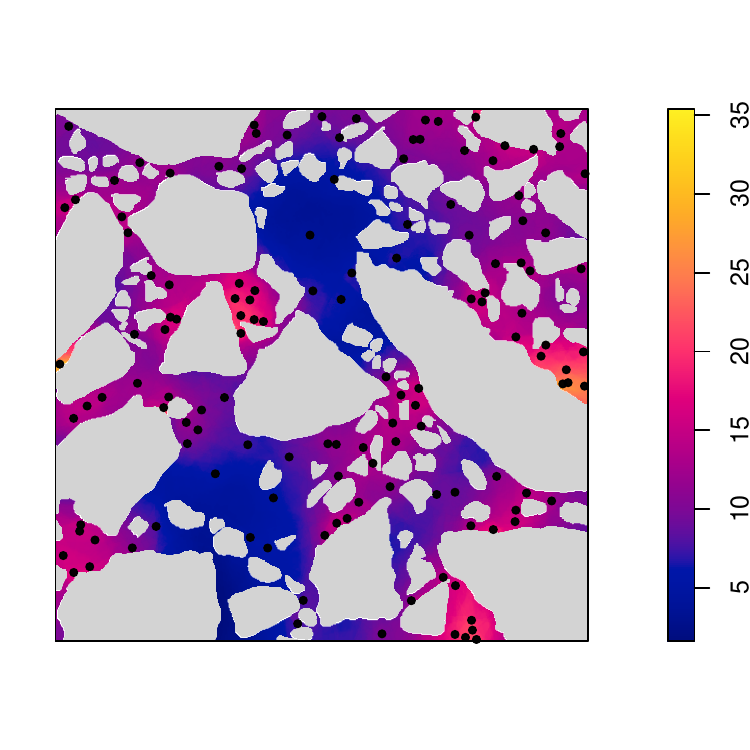} \end{tabular}
   \caption{Left : locations of 136 centroids of air bubbles in a 5.6 mm$^2$ cross-section of a concrete body. The grey zones correspond to aggregate particles and the white zone is the cement paste matrix where the bubbles are located. Middle : a random tessellation with approximately 40 cells coloured by the intensity of points in each cell. Right : random forest intensity estimation based on 100 independent tessellations.}
          \label{fig:concrete}
\end{figure}

We illustrate our approach to intensity estimation on a manifold in Figure~\ref{fig:bei_mesh}. It shows the log-intensity of Beilschmiedia pendula
        Lauraceae trees in a $1000m \times 500m$ region of  Barro Colorado Island, accounting for the topography of the region. This dataset, originally studied in \cite{condit1998tropical,hubbell1999light}, is also available in the \texttt{R} package \texttt{spatstat} and has been extensively analysed in several articles. In particular, a kernel-based intensity estimation accounting for the topography is  investigated in \cite{Cohen23}: it requires to compute geodesic distances on the manifold in addition to edge and shape corrections. In comparison, our method is straightforward and boils down to counting the number of points in the cells of each generated tessellation of the manifold.

\begin{figure}
\begin{tabular}{cc}            
 \includegraphics[width=0.5\textwidth]{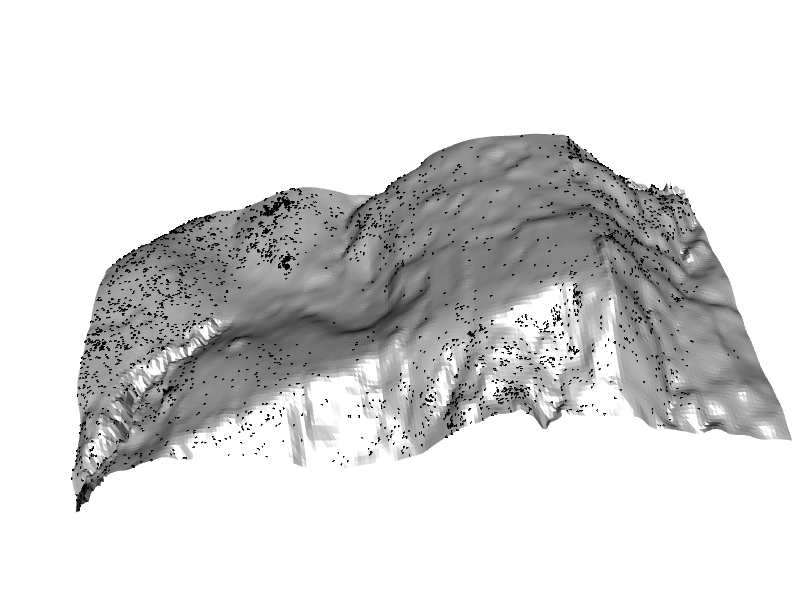} &
 \includegraphics[width=0.5\textwidth]{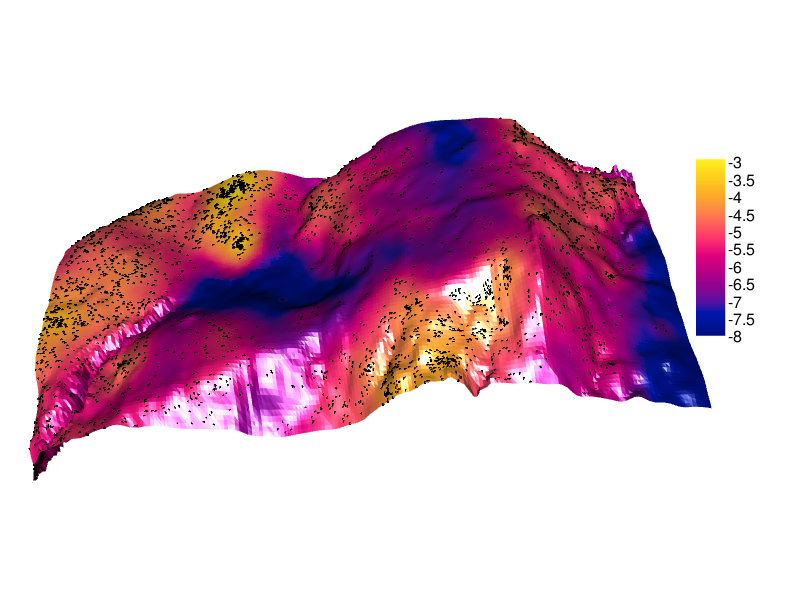}\\[-6ex]
  \includegraphics[width=0.5\textwidth]{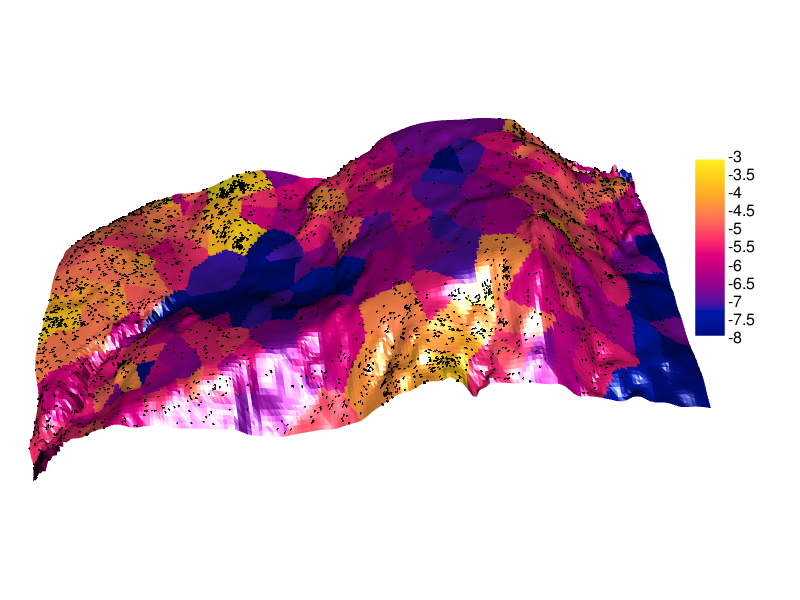} &
   \includegraphics[width=0.5\textwidth]{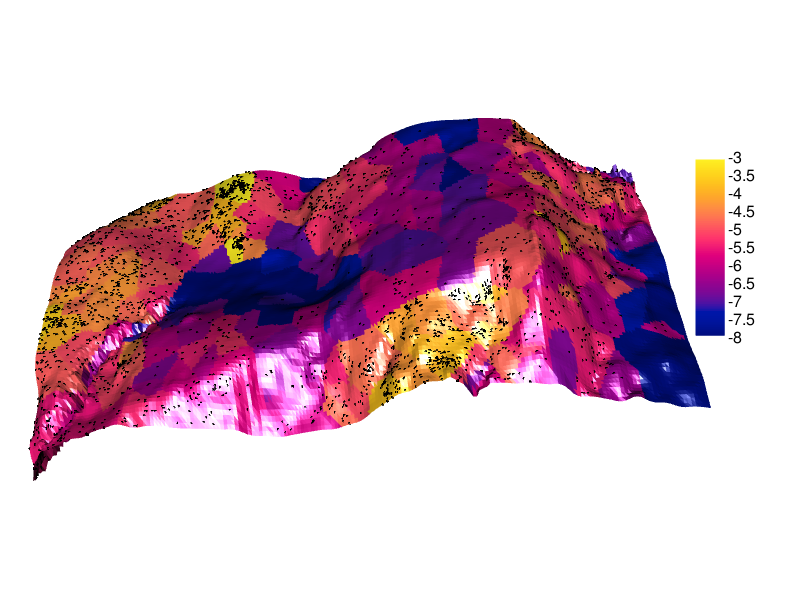}  
 \end{tabular}
   \caption{Top left: locations of $3604$ Beilschmiedia pendula
        Lauraceae tree stems observed in a $1000m \times 500m$ region in  Barro Colorado Island. Top right: logarithm of the random forest intensity estimation based on 100 independent tessellations of the surface, each with 235 cells on average. Bottom: two random tessellations coloured by the log-intensity of points in each cell.}
          \label{fig:bei_mesh}
\end{figure}

%

\section{Partitions based on covariates}
\label{sec:covariatespart}

As 
proved in
Section~\ref{sec:benefits covariates}, leveraging covariates, when available, generally improves
the rate of convergence of the intensity estimator, compared to 
 the purely spatial case where no covariates are taken into account. 
  When some covariates $z=(z_1,\dots,z_p)$ are available, we construct the random forest intensity estimator by following the same basic steps as for a random forest regressor.   We detail below the splitting procedure for each tree in the forest, describe how to tune the parameters using out-of-bag cross-validation, and provide an illustration using the Bei dataset.
  
  \subsection{Splitting rule and variable importance}
  
The full construction of the random forest intensity estimator in presence of covariates is summarised in   Algorithm~\ref{alg:random forest intensity}. It consists, for  each intensity tree estimator, in partitioning $W$ through a recursive partition of $z(W)$.

 \begin{algorithm}
        \caption{Random forest intensity estimator based on
        covariates}\label{alg:random forest intensity}
        \begin{algorithmic}[1]
          \Require Point pattern $X$ observed on $W$;
          covariates $z$; parameters $M, mtry, minvol$.
          
          \For{$i=1,\ldots, M$} \label{alg:line:forloopbegin}
      
            \State Draw a bootstrap version $X_b$ of
            $X$, with replacement;\label{alg:line:bootstrap}
      
            \State Compute a partition $\{A_{j}, j\in \mathcal
            J\}$ of $W$ as follows;

            \State Let $W$ be the unique cell in the initial partition;

	   \For{each cell with a volume greater than $minvol$ and having at least 1 point from $X_b$}

            \State Pick $mtry$ covariates at random;

                \For{each picked covariate $z_k$}
                               \State Compute $\bar{z}_k$, the median of $z_k$ in  the cell;

                \State Deduce the sub- and super-level sets of $z_k$ in the cell, w.r.t $\bar{z}_k$;

                \State Compute the splitting score of the cell for $z_k$ as in Equation~\eqref{e:score
                split};\label{alg:line:splitchoice}

                \EndFor
                
                  \State Select the covariate leading to the maximal splitting score;

               \State Split the cell into the sub- and super-level sets of the covariate; 

		\State Update the partition with this split;

 \EndFor
                     
            \State Compute the tree intensity estimate
            $\hat\lambda^{(i)}$ as in Equation~\eqref{e:tree
            estimator without n};
            
            \EndFor
         
            \State \textbf{Output:} The random forest intensity estimator:
            $\frac{1}{M} \sum_{i=1}^M \hat\lambda^{(i)}$. \label{alg:line:output}
        
                \end{algorithmic}
        
\end{algorithm}

  Before detailing the construction of each tree, let us first outline the two standard strategies we use to introduce randomness, and specifically diversity, among them---a key ingredient in the performance of random forests. 
  First, each tree is built from a bootstrap sample $X_b$ of $X$, obtained by drawing $n$ points from $X$ with replacement, where $n$ is the cardinality of $X$. This step is discussed in Remark~\ref{rem comments} below. Second, at each node, the splitting rule is based on a randomly selected subset of $mtry$ covariates out of the $p$ available, where $mtry$ is a tuning parameter.  
   
 Now, in  growing a tree based on $X_b$, the key specificity of our point process setting lies in how a cell of the partition is (or is not) split. To this end, we first compute, for each covariate $z_k$ among the $mtry$ selected covariates, its median value $\bar z_k$. Then we consider the sub- and super-level sets of $z_k$ in the cell with respect to $\bar z_k$,  denoted by $L_{z_k}^{-}$ and $L_{z_k}^{+}$, respectively. For instance, if $A_j$ denotes the candidate cell to be split, $L_{z_k}^{-}=\{u\in A_j, z_k(u)\leq \bar z_k\}$. Then we compute the following splitting score
   \begin{equation}\label{e:score split}
       s(z_k)= n_-\log\left( \frac{n_--1}{|L_{z_k}^{-}|} \right)  \1_{n_->1} 
        + n_+ \log\left( \frac{n_+-1}{|L_{z_k}^{+}|} \right) \1_{n_+>1},
\end{equation}
where $n_- = |X_b \cap L_{z_k}^{-}|$ and  $n_+ = |X_b \cap
L_{z_k}^{+}|$. This specific form is related to the variation in the leave-one-out Poisson log-likelihood caused by the split, and is further justified in Remark~\ref{rem comments} below. Given this, the tree construction is straightforward: we split each cell according to the sub- and super-level sets of the covariate associated to the highest score. 
The procedure is repeated until a stopping criterion is met. This criterion may be based on either a minimum number of points per cell or a minimum cell volume. In Algorithm~\ref{alg:random forest intensity}, we choose to split every non-empty cell whose volume exceeds a predetermined threshold, denoted by $minvol$.

Note that the splitting position is chosen as the median of the covariate values, as in median trees (\cite[Chapter 20.3]{Devroye2013}, \cite{Biau_Devroye13}) and median forests \citep{klusowski21b,duroux18} in a regression setting. For categorical covariates, following a standard trick, we first assign to each category its corresponding intensity estimate, defined as the number of points in that category divided by the volume of its associated level set in $W$. This provides a numerical representation of each categorical covariate, so that the same median-based splitting rule can be applied to both numerical and categorical covariates. As an extension, the splitting position could be optimised over several quantile values of the covariates.
However, this generalisation substantially increases computational time while providing, in our experience, no significant improvement.

The above procedure leads to a partition $\{A_{j}, j\in \mathcal
            J\}$ of $W$ and then to the intensity tree estimator \eqref{e:tree
            estimator without n}.  The random forest intensity estimator \eqref{e:rf estimator without n} is finally obtained by generating $M$ tree estimators, each based on an independent realisation of $X_b$.

                 A  by-product of the random forest approach is that we can quantify variable importances (VIMP).  A natural approach in the context of random forests is to measure the improvement in the splitting score at each split attributable to the selected variable. In our setting this becomes the gain in the leave-one-out Poisson log-likelihood score $LCV$, defined in equation~\eqref{LCV2} below. Specifically, if the cell $A_j$ in a given tree is  split using $z_k$, this gain  is 
                  $$ \mathrm{VIMP}(z_k|A_j)=  s(z_k) -  n_j\log\left( \frac{n_j-1}{|A_j|} \right)  \1_{n_j>1},$$
        where $n_j= |X_b \cap A_j|$ and $s(z_k)$ is given by \eqref{e:score split}. If the split is not due to $z_k$, then $\mathrm{VIMP}(z_k|A_j)=0$. This includes the case where $z_k$ is not among the randomly selected candidate covariates at that node. The  variable importance of $z_k$ for a given tree  is then obtained by summing these gains over all splits using $z_k$ in that tree. The overall VIMP of $z_k$ in the random forest is then simply the average of its importance over the $M$ trees.


  \begin{remark}\label{rem comments}
  The previous construction calls for two comments. The first one concerns the bootstrap sample $X_b$.  A reader familiar with point processes might be unsettled by the fact that $X_b$ contains multiple points, due to the replacement step. While it would cause a problem if we were interested in studying the cross-dependencies between the points of $X$, it is not a concern as long as we focus solely on the intensity. Indeed, the expected number of points of $X_b$ in any subregion is equal to that of $X$.   
 Note that, alternatively to a bootstrap sample, we could use a thinning realisation of $X$, as considered in \cite{cronie23}. This option is available in our R package \texttt{spforest}. However, in our experience, it does not improve estimation accuracy, but requires an additional parameter, namely the thinning probability, and increases the computational cost.
  
            The second comment is about the choice of the splitting score \eqref{e:score split}. A standard procedure to conduct parametric estimation of the intensity of a point process is by maximising the Poisson likelihood. This approach is not only consistent for genuine Poisson point processes, but it also makes sense for a much wider class of point process models,  see  \cite{Schoenberg05,waagepetersen:guan:09}. For intensity kernel estimation, the Poisson likelihood is also used as a cross-validation score to choose the bandwidth, see~\cite{baddeley15book}.  In this case the leave-one-out version of the Poisson log-likelihood is employed. It is defined by 
            \begin{equation}\label{LCV} 
            LCV=\sum_{x\in X} \log \hat \lambda_{-x}(x) - \int_W \hat\lambda(u)\der u,\end{equation}
            where  $\hat \lambda_{-x}(x)$ denotes the estimation of $\lambda(x)$ without using the event $x$. Following this idea, we use this score to quantify the relevance of a split in our tree construction. For a partition $\{A_{j}, j\in \mathcal  J\}$ and the associated estimator  \eqref{e:tree
            estimator without n}, $LCV$ reads 
            \begin{equation}\label{LCV2} LCV=\sum_{j\in \mathcal  J} \left(n_j \log \frac{n_j-1}{|A_j|}  \1_{n_j>1} - n_j\right),\end{equation}
where $n_j$ denotes the number of events in $A_j$. 
The score \eqref{e:score split} then corresponds to the contribution of the split of $A_j$ due to $z_k$ to the total $LCV$. Maximising  \eqref{e:score split} over all covariates amounts to maximise the increase in $LCV$ in  the split of $A_j$. Note that this score does not involve any integral term, so that no numerical approximation is required. 
\end{remark}

\subsection{Out-of-bag cross-validation}
\label{sec:OOB}

When fitting a random forest as in Algorithm~\ref{alg:random forest intensity} we need to specify the three tuning parameters $mtry$, $minvol$ and $M$.  This section shows how we can adapt  the out-of-bag (OOB) cross-validation procedure to our setting in order to conduct these choices.  

OOB cross-validation is a standard approach for  random forests. For each tree, the OOB sample is $X\setminus X_b$. Following our choice for the splitting score \eqref{e:score split}, motivated in  Remark~\ref{rem comments}, we can assess the quality of estimation of a tree on the OOB sample,  through the Poisson log-likelihood score. Accordingly,  if  $\hat\lambda^{(i)}$ denotes the tree intensity estimator  and $X_b^{(i)}$ is the bootstrap sample for this tree, we call OOB score in our setting the quantity 
\begin{equation}\label{e:OOB_score}
        OOB_{i} = \sum_{x \in X \setminus X_b^{(i)}} 
                \log(\hat\lambda^{(i)}(x)).
\end{equation}
Note that this is the score given by \eqref{LCV}, except that the leave-one-out step is not necessary
here and  the integral term, which  equals $|X_b\cap W|=|X\cap W|$ and does not depend on the hyperparameters, has been removed.
The OOB score of the random forest estimator is
$$OOB=\frac 1 M \sum_{i=1}^M  OOB_{i}.$$
The tuning parameters are then chosen by maximising $OOB$. We show in our simulation study of Section~\ref{sec:simulation} that this provides a good choice, in accordance with the optimal (but unknown) oracle choice based on the minimal mean integrated squared error of the intensity estimator.

In some cases, however, OOB may fail, in the sense that it might be monotone with respect to one of the hyperparameters. In our experience, this may happen with highly clustered point patterns, where the points of $X_b$ and $X\setminus X_b$ are essentially located at the same places, which prevents a proper cross-validation evaluation. In such cases, the smaller $minvol$, the smaller the OOB, which is clearly not appropriate. If OOB fails in this way, we suggest setting $minvol$ as the mean volume $\gamma^{-1}$ induced by the spatial rule of thumb in \eqref{choice gamma}.

\subsection{Illustrations with covariates}\label{section:bei}
We apply our methodology to the Bei dataset, already considered in Figure~\ref{fig:bei_mesh}, which records the locations of trees in a $1000m \times 500m$ region. This dataset includes in fact 15 covariates describing  topography and soil
composition attributes, namely 
\begin{equation}
        \label{e:sim list covariates}
        elev, grad, Al, B, Ca, Cu, Fe, K, Mg, Mn, P, Zn, N, N_{min}, pH.
\end{equation}
We consider the intensity estimation leveraging the 15 covariates listed above. 
The OOB cross-validation procedure resulted in the choice of hyperparameters $mtry=12$, $minvol = 250$ and $M=500$. The estimated log-intensity is shown in Figure~\ref{fig:bei PP}.

\begin{figure}
        \begin{center}
                \includegraphics[width=\textwidth]{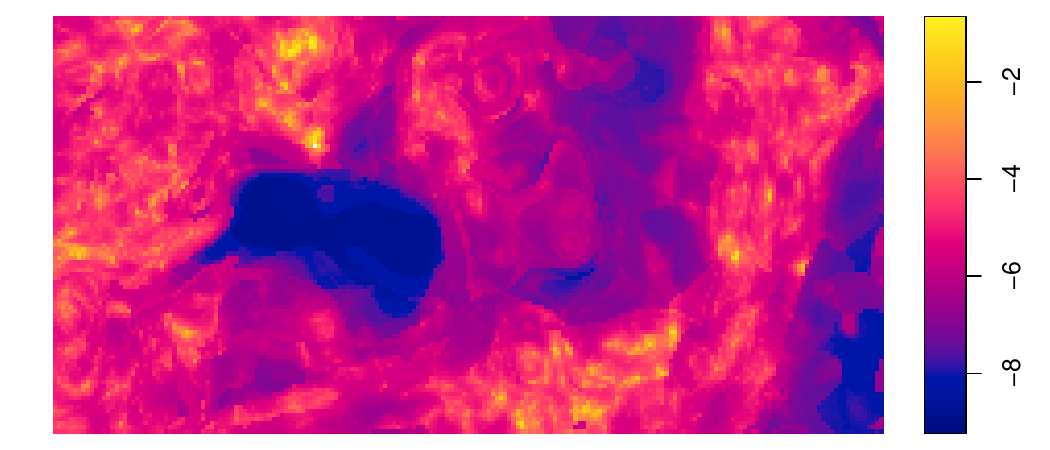} 
        \end{center}
        \caption{Estimated log-Intensity of the Bei
        trees by random forest.}
        \label{fig:bei PP}
\end{figure}

\begin{figure} 
        \begin{center}
                \includegraphics[width=\textwidth]{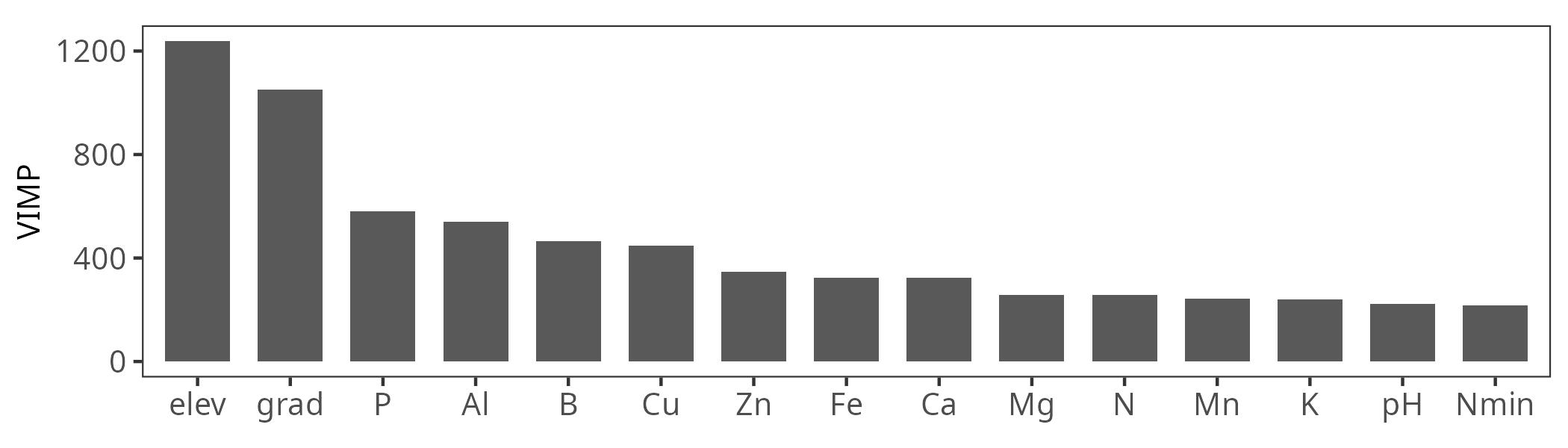} 
        \end{center}
        \caption{VIMP of each covariate contributing to the estimated intensity of the Bei dataset of Figure~\ref{fig:bei PP}.}
        \label{fig:bei vimp}
\end{figure}


The variable importance of each covariate is displayed in Figure~\ref{fig:bei vimp}, showing that the topographic variables $elev$ and $grad$ are the two most important ones, followed, to a lesser extent, by some soil nutrients. It is interesting to compare these findings with previous studies on the Bei dataset, such as \cite{waagepeterne:07}, \cite{waagepetersen:guan:09}, and \cite{choiruddin18}, where a log-linear parametric model was fitted, and \cite{Lu06082025}, where the nonparametric gradient boosting method XGBoostPP was used. These studies also identified the topographic variables as important. But unlike our results, they found one soil nutrient to be as important as the topographic variables, although they disagree on which one: $P$ in \cite{choiruddin18} and $Cu$ in  \cite{Lu06082025}. These two variables are also important in our study, but in a lesser extent. Nevertheless, the potentially strong correlations between soil nutrients may complicate the identification of the most important covariates, and the results should therefore be interpreted with caution. Since the ground truth is unknown for this real-world dataset, it is not possible to determine which of the findings  is the most reliable.


\section{Numerical study}
\label{sec:simulation}

We conduct in this section several simulation studies to evaluate our approach. Section~\ref{synthetic} uses a synthetic model to assess the suitability of the OOB cross-validation procedure, the estimation accuracy of our method, and the performance of the VIMP scores. Section~\ref{sec:XY} investigates whether spatial coordinates should still be included as covariates when exogenous covariates are available.
In Section~\ref{xgboost}, we compare our approach with state-of-the-art methods for intensity estimation in the presence of covariates, across different scenarios. Finally, Section~\ref{sec:comp} discusses the computational complexity of our method. 

\subsection{Synthetic Datasets}\label{synthetic}

We use in this section the covariates available in the Bei dataset (see Section~\ref{section:bei}) to generate a synthetic Poisson point process on
$W=1000m \times 500m$ with intensity:
\begin{equation}\label{e:sim2 true intensity}
\lambda(x)=c\, \exp\left(0.5 \, \psi(Mn) (x) + 1.2 \,  \tilde Zn(x) + 0.8 \, \tilde Fe(x)\right).
\end{equation}
Here $\tilde Zn(x)$ and $\tilde Fe(x)$ denote the covariates $Zn$ and $Fe$, normalized to $[0,1]$, at location $x\in W$, and $\psi(Mn) (x)=1 + \sin(20+Mn(x)/ 100)$ is a non-linear transformation of $Mn$. The normalising constant $c$ was chosen so as to yield approximately 1000 points in $W$, which corresponds to $c\approx 6.19\times 10^{-4}$. 
This intensity is shown in the top left plot of Figure~\ref{fig:non linear image}. 
The estimation of $\lambda$ for this synthetic dataset is based on the $p=15$ covariates listed in \eqref{e:sim list covariates}, even if only three of them are relevant. 

In Figure~\ref{fig:non linear oob}, we first assess the appropriateness of the OOB cross-validation procedure for selecting  the hyperparameters $mtry$ and $minvol$ of the random forest estimator, as described in Section~\ref{sec:OOB}. For this illustration, the number of trees is fixed at $M=300$. The left-hand plot displays  the OOB score for different values of $mtry$ (shown as separate curves) as a function of $minvol$ (in proportion of the total volume of $W$), averaged over 100 replications. The right-hand plot shows the same representation for the negative of the MISE, which is of course unknown in practice and can be viewed as an oracle score for hyperparameter selection. The consistency between the two plots demonstrates the suitability of using the OOB score in practice. Note that the curves have a stepwise shape, reflecting the discrete number of points within each cell. This suggests some robustness in the choice of the optimal value of $minvol$. This case study also suggests that the standard choice of $mtry$ in regression settings, namely  $mtry\approx p/3$, seems to provide good results, whereas the choice of $minvol$ is crucial. As to the number of trees $M$, of course the larger the better, but we noticed in our experiments that performance  stabilizes for $M>200$.

\begin{figure}
        \begin{center}
  \includegraphics[width=0.9\textwidth]{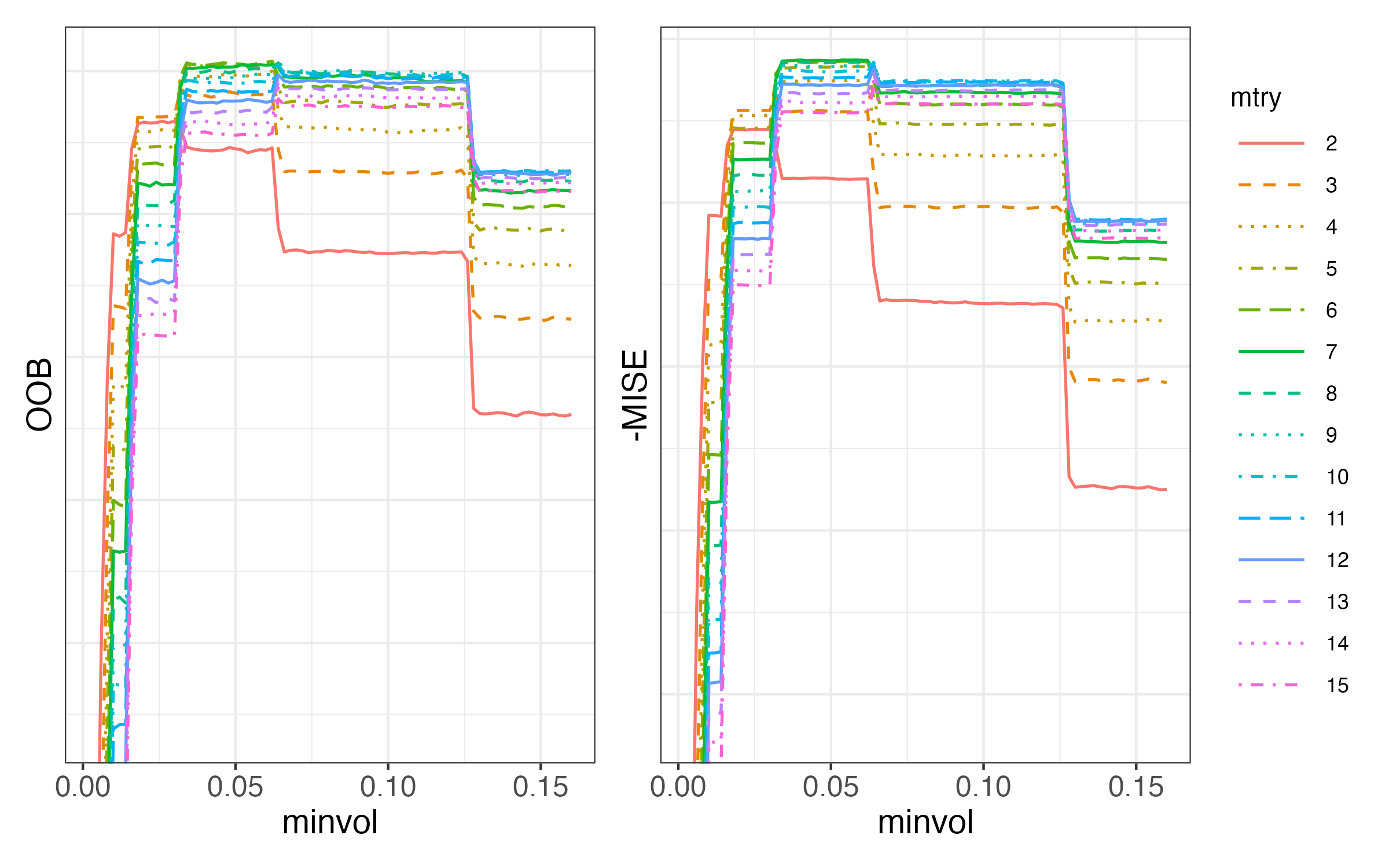}
    \end{center}
        \caption{Left: OOB score for different values of $mtry$ as a function of $minvol$, averaged over 100 realisations of the model in Section~\ref{synthetic}. Right: same representation for the negative of the MISE.}
        \label{fig:non linear oob}
\end{figure}

Figure~\ref{fig:non linear image} displays four plots obtained from one realisation of the above model: the true intensity \eqref{e:sim2 true intensity} (top-left); the estimation based on the random forest estimator using all 15 covariates listed in \eqref{e:sim list covariates},  with hyperparameters selected by OOB cross-validation (top-right); the estimation obtained without covariates, using only the spatial coordinates and following the procedure of Section~\ref{sec:purelyspatialpart} (bottom-left); and the parametric estimation under the misspecified log-linear model, as implemented by the function \texttt{ppm} of the R package \texttt{spatstat} (bottom-right).
These plots show that the random forest estimator with covariates captures well the behaviour of the true intensity, while the purely spatial estimator is less accurate. This illustrates the benefit of using covariates when available, as further confirmed in Section~\ref{sec:benefits covariates}. In addition, the misspecified log-linear parametric model appears inappropriate. These visual impressions are supported by numerical results based on the MISE over 100 replications, which are $0.108$ for the random forest estimator, $0.174$ for the purely spatial estimator  and $0.220$ for the log-linear model.

\begin{figure}
        \begin{center}
                \includegraphics[width=0.49\textwidth]{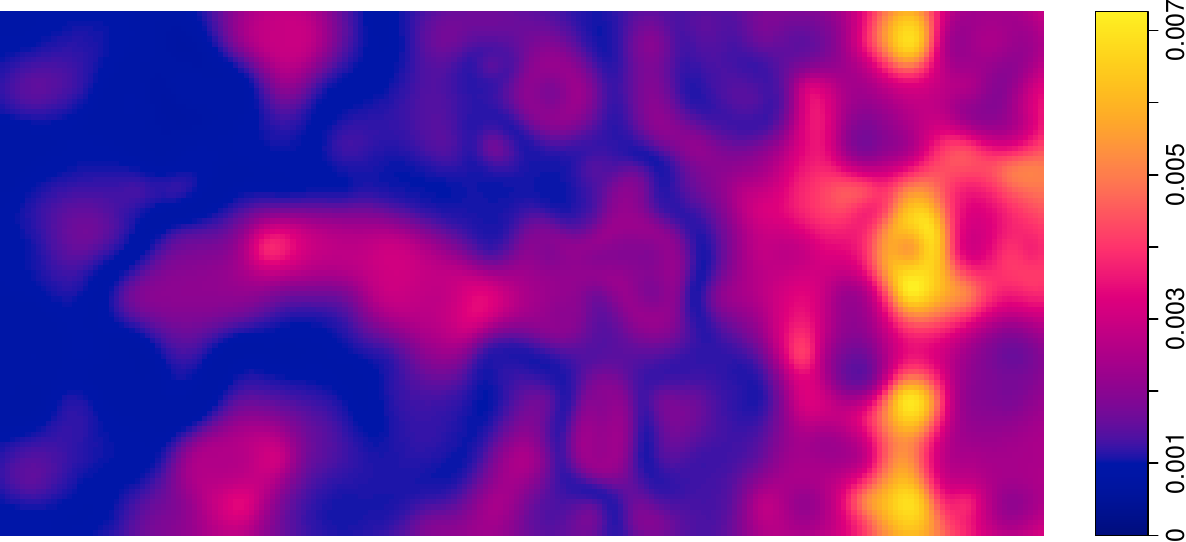}  
                \includegraphics[width=0.49\textwidth]{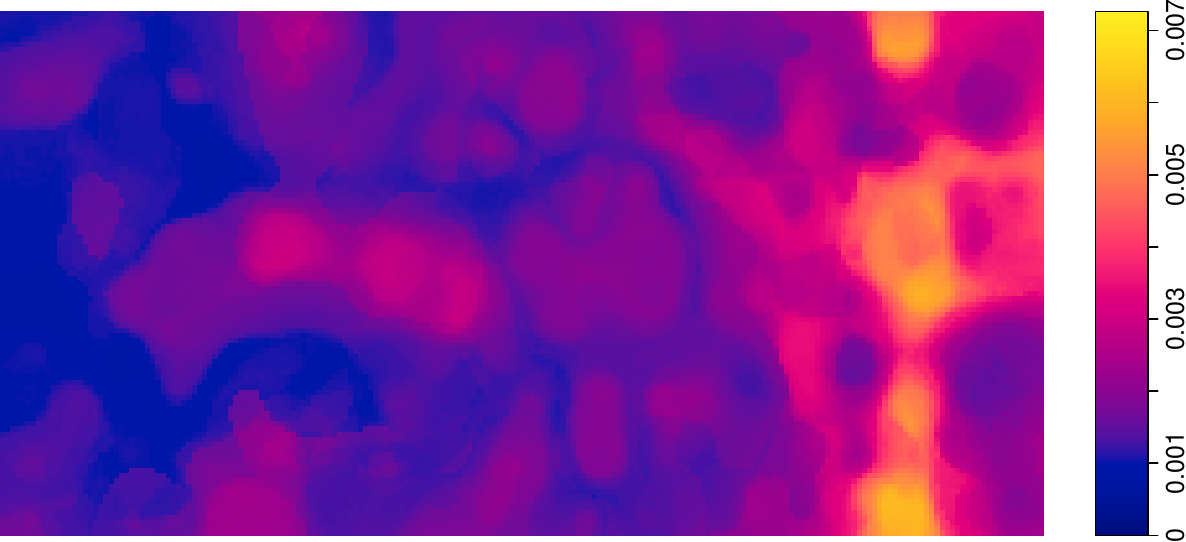} \\ 
                \includegraphics[width=0.49\textwidth]{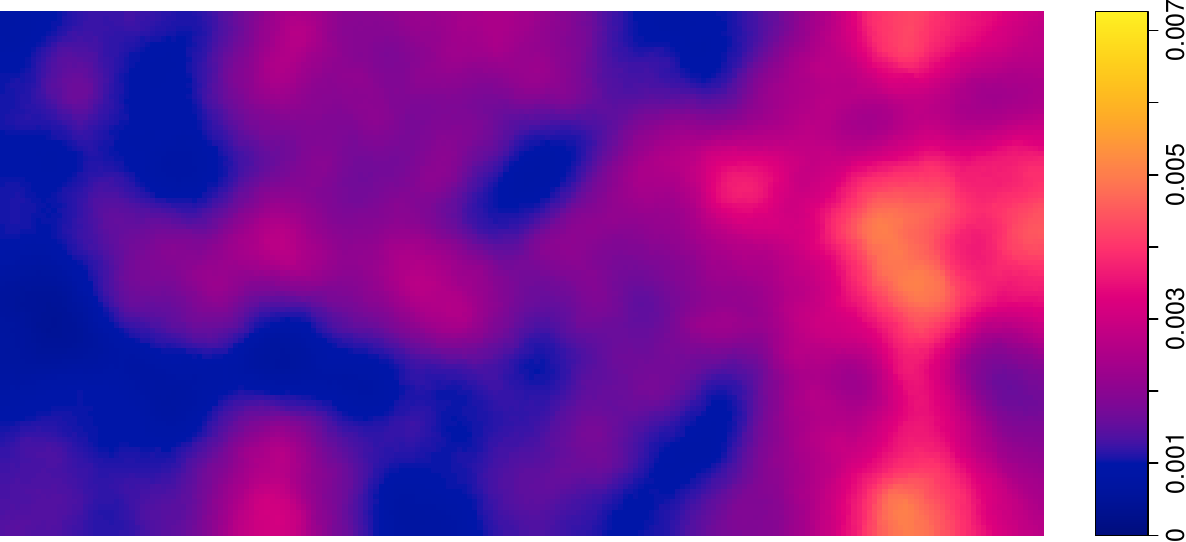}
                \includegraphics[width=0.49\textwidth]{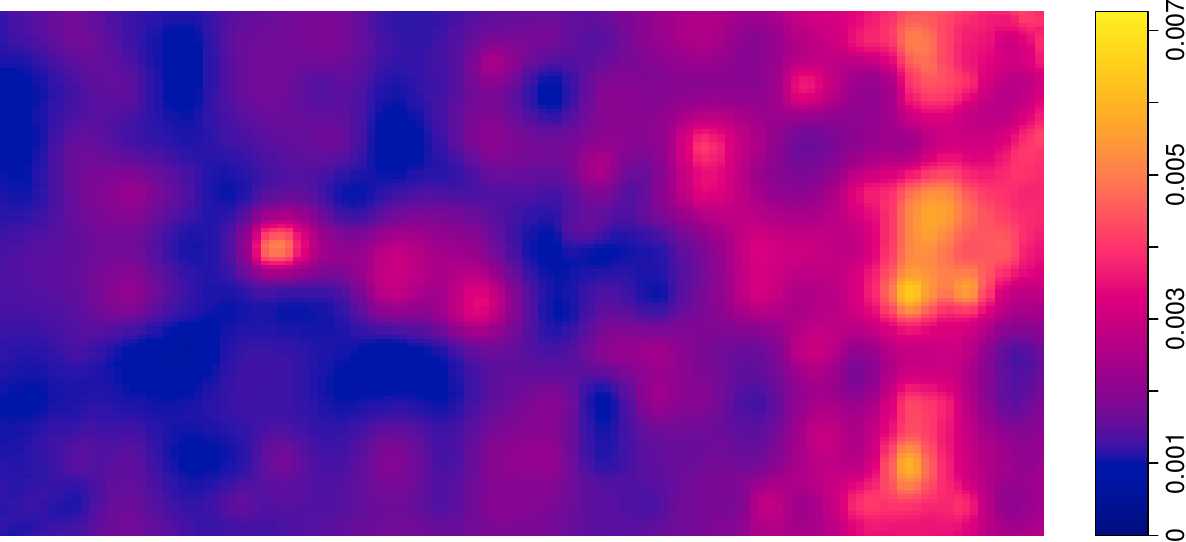}  
        \end{center}
        \caption{Top left: true intensity as defined
        by~\eqref{e:sim2 true intensity}. Top right: 
        intensity estimation by random forest using the $p=15$ available covariates, based on one realisation. Bottom left: estimation without using the covariates, as in Section~\ref{sec:purelyspatialpart}.   Bottom right:  estimation using a (misspecified) parametric log-linear model. }
        \label{fig:non linear image}
\end{figure}

Finally, Figure~\ref{fig:non linear vimp} reports the importance (VIMP) of each covariate over 100 replications, when using the random forest estimator with all covariates. It clearly identifies the three relevant covariates used in the model, namely $Fe$, $Mn$ and $Zn$.
In contrast, for the same simulations, the misspecified log-linear model detects the significance of $Zn$ and $Fe$ in most cases (88 and 100 out of 100 replications, respectively, by a Wald test at the 5\% level), which is expected since these covariates appear log-linearly in $\lambda$. But it generally fails to detect $M_n$ (only 16 times), while $Al$ and $Cu$, which are not relevant in the model, are found significant in $54\%$ and $56\%$ of the cases.

\begin{figure}
        \begin{center}
                \includegraphics[width=\textwidth]{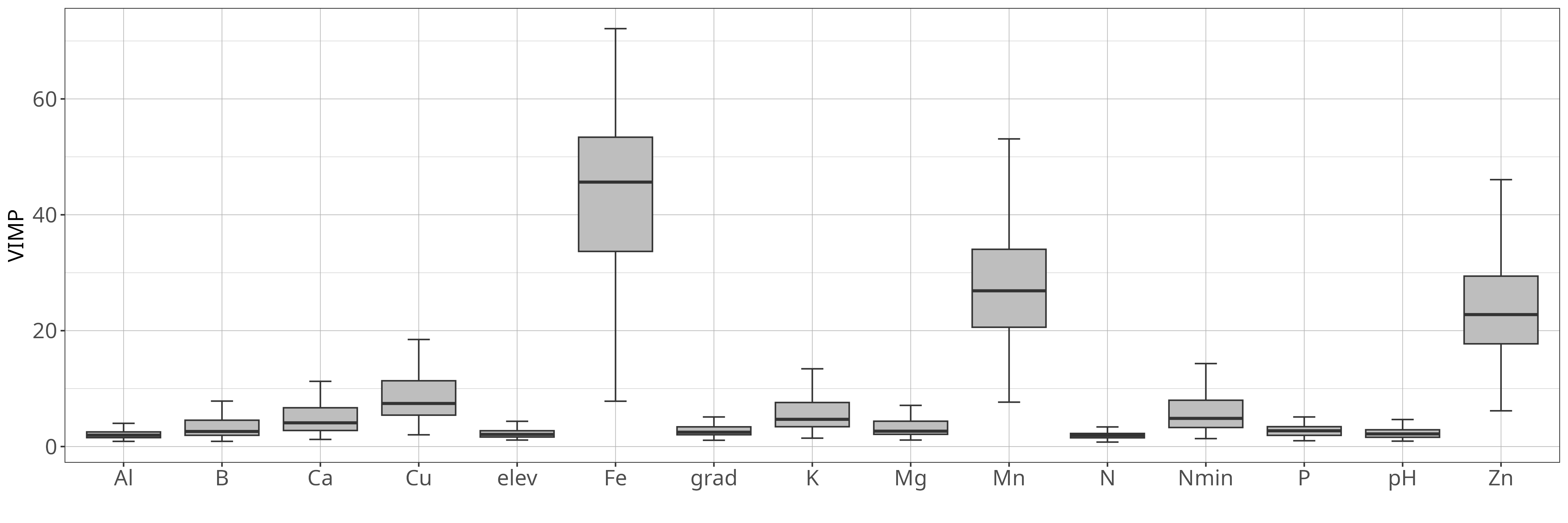} 

        \end{center}
        \caption{Boxplot of the VIMP of each covariate over 100 replications of the synthetic model described in  Section~\ref{synthetic}, that truly depends non-linearly on $Fe$, $Mn$ and $Zn$.}
        \label{fig:non linear vimp}
\end{figure}

\subsection{Including both spatial coordinates and covariates}\label{sec:XY}

  In the presence of covariates, a natural question is whether they are sufficient to explain the intensity. A cautious approach in practice can be to include the spatial coordinates alongside the covariates in the estimation procedure.
Through a short simulation study, we investigate whether this inclusion can improve estimation accuracy when the covariates alone are insufficient to explain the intensity. We also assess whether including the spatial coordinates deteriorates estimation accuracy when they are not needed. Finally, we analyse the VIMP scores to determine whether they can still identify the relevant covariates when  the spatial coordinates are included.

To conduct this study, we generate eight independent realisations of a Gaussian field on the unit square, with mean 0, variance 0.5 and having an exponential correlation function  with scale
parameter $0.5$. Our eight covariates $Z_1,\ldots,Z_8$ consist of the exponential values of these realisations, normalised between 0 and 1. We then consider four Poisson point processes on the unit square, with respective intensity 
\begin{equation*}
\begin{array}{ll}
\lambda_1 = c (Z_1+Z_2), & \quad \lambda_2= c (Z_1 + x / 2),\\
\lambda_3= c \left( 7 Z_1  Z_2 + Z_3 + Z_4 \right),&  \quad \lambda_4 = c \left(2 x Z_2 + Z_3 + Z_4\right),\end{array}
\end{equation*}
where $x$ denotes the $x$-axis coordinate of the unit square. In each model, the constant $c$ is chosen so that, on average, there are 500 points.

Over 100 replications of each of the four models described above, we estimate the intensity by including all eight covariates along with the spatial coordinates $x$ and $y$, using the methodology described in Section~\ref{sec:covariatespart}. The variable importances, averaged over the 100 replications, are reported in Figure~\ref{fig:vimpXY}.
In addition, for each model, we compare the performance of intensity estimation with and without the spatial coordinates. The relative excess errors are reported in Table~\ref{tab:error}. They measure the relative increase in the mean squared error, averaged over the 100 replications, resulting from the inclusion of the spatial coordinates.

\begin{figure}
    \begin{center}
        \includegraphics[width=0.24\textwidth]{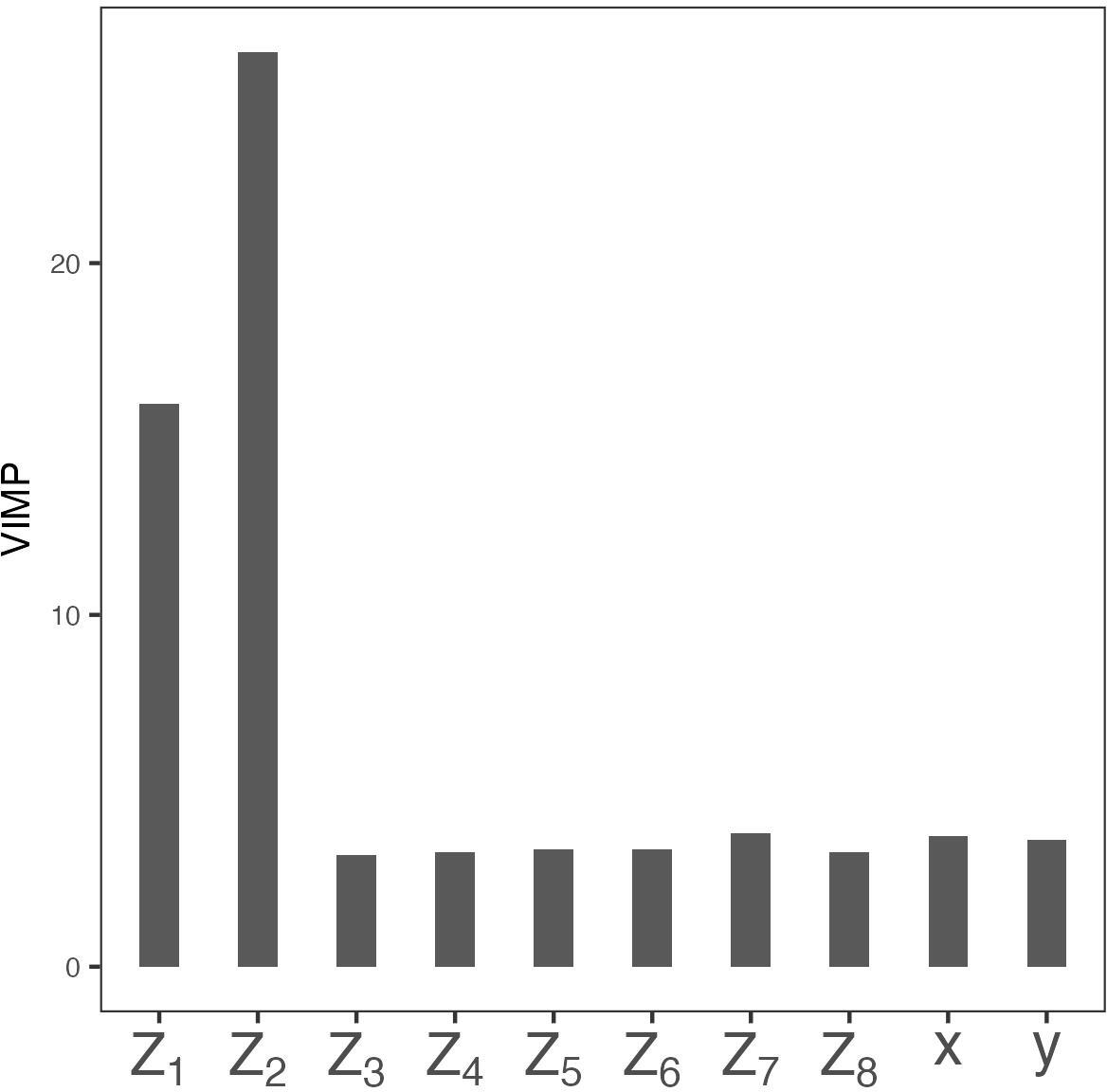}  
        \includegraphics[width=0.24\textwidth]{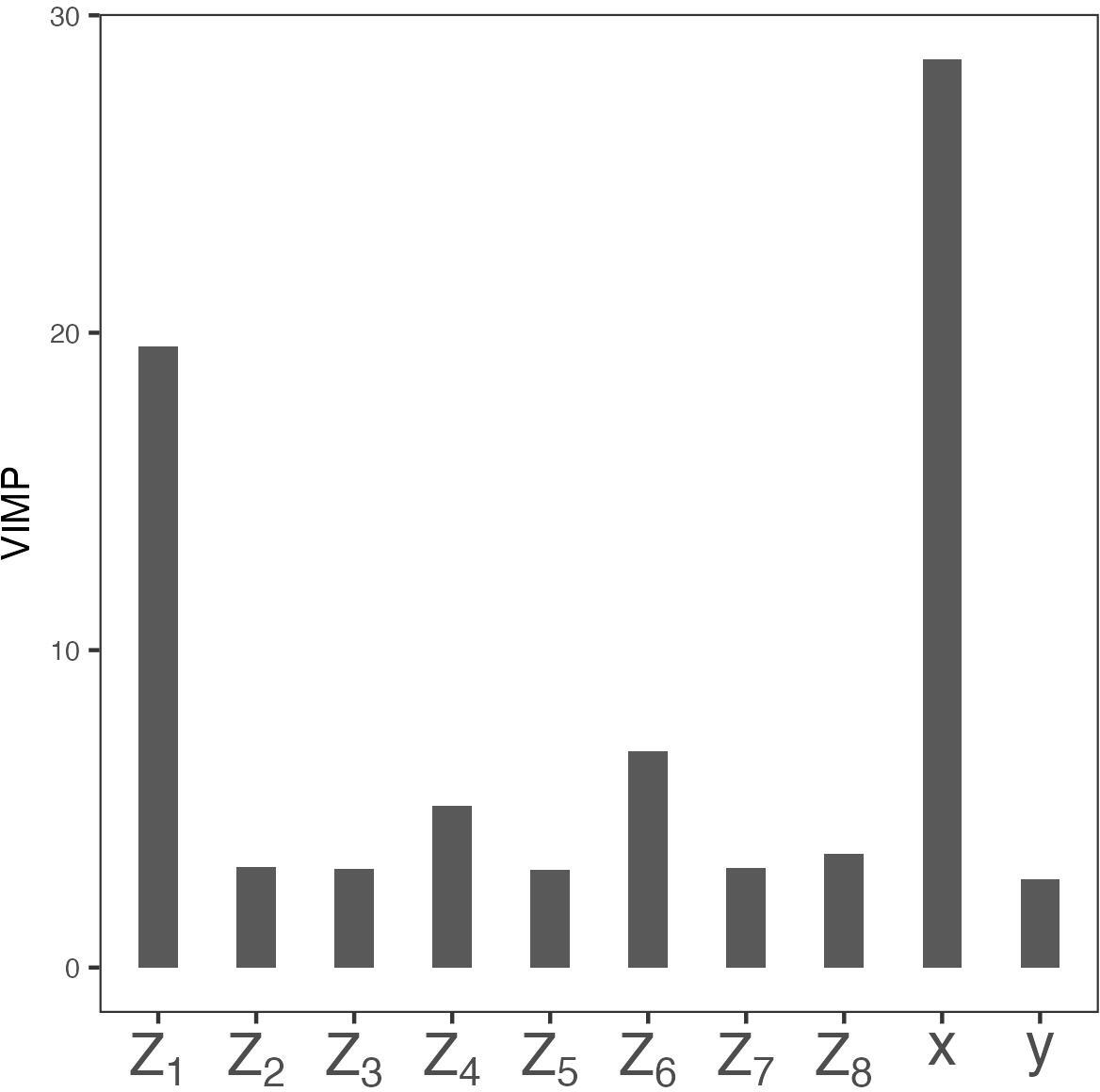} 
          \includegraphics[width=0.24\textwidth]{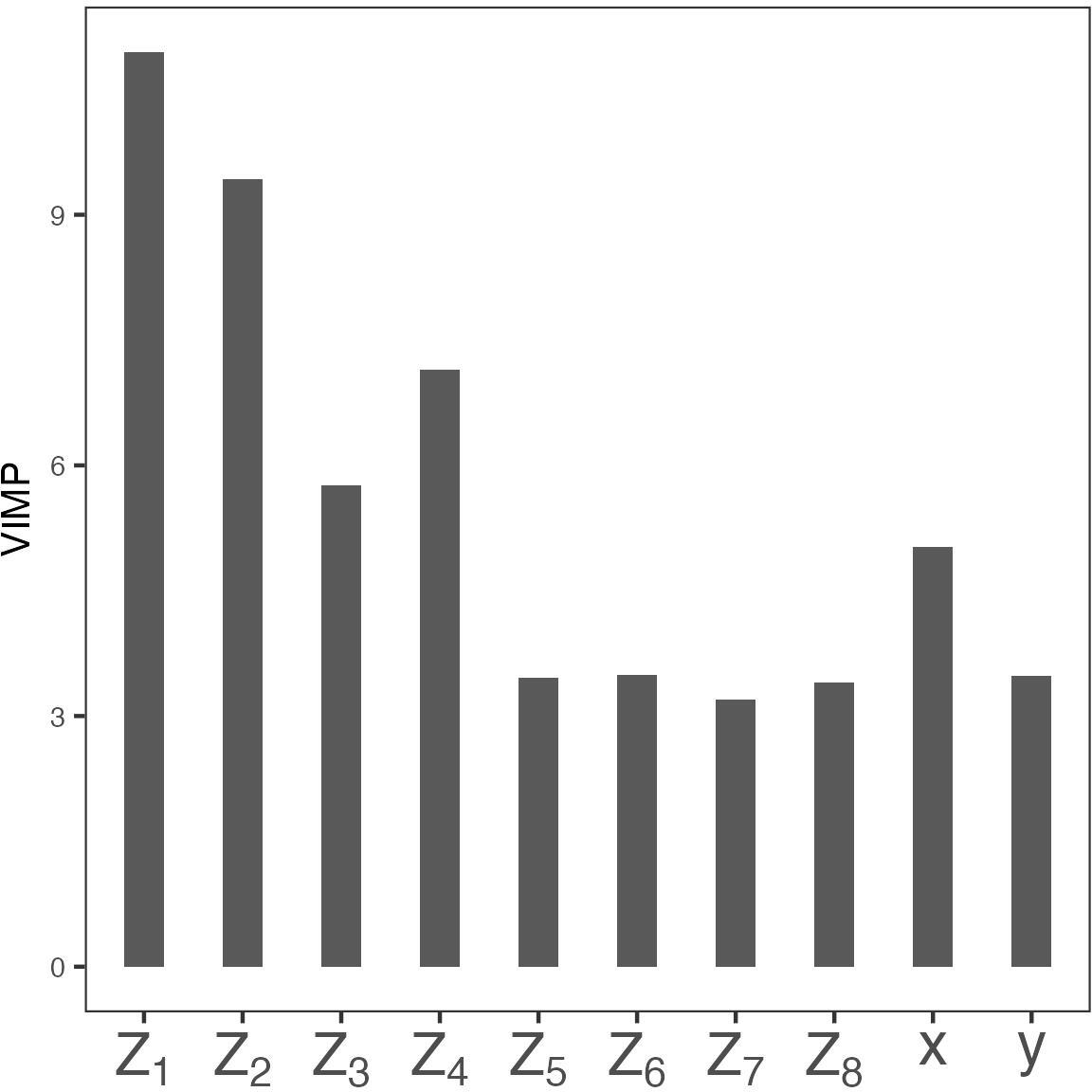}  
        \includegraphics[width=0.24\textwidth]{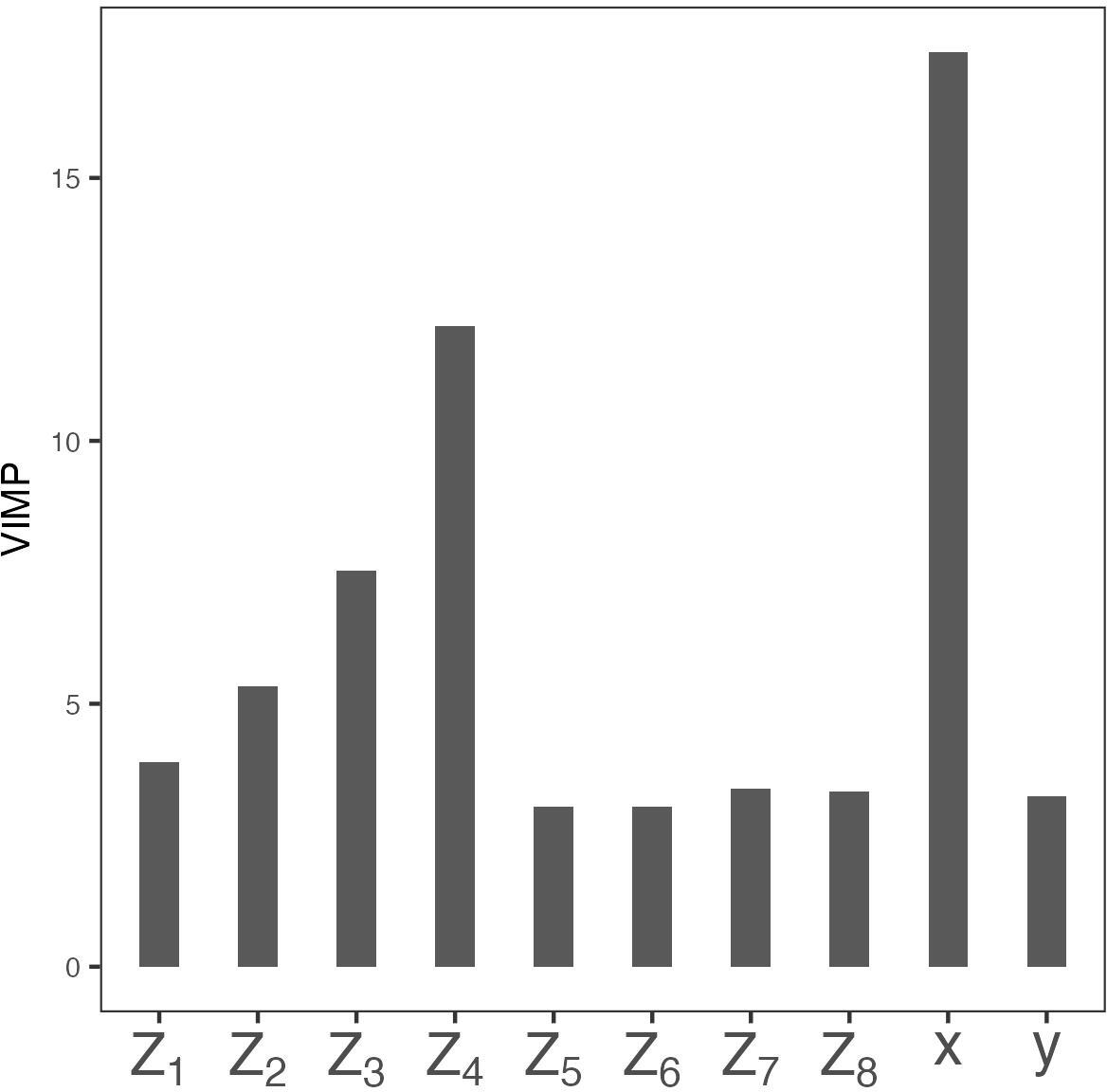} 
    \end{center}
    \caption{VIMP averaged over 100 replications of a Poisson point process with respective intensity $\lambda_1,\dots,\lambda_4$, which in theory depend on $(Z_1, Z_2)$, $(Z_1,x)$, $(Z_1,Z_2,Z_3,Z_4)$ and  $(Z_2,Z_3,Z_4,x)$, respectively.}
    \label{fig:vimpXY}
\end{figure}

\begin{table}
\caption{Increase in  mean squared errors of intensity estimation when including spatial coordinates, averaged over 100 realisations of the models with intensities $\lambda_1,\dots,\lambda_4$. Standard deviations of estimation are given in parentheses.}
\label{tab:error}
\begin{tabular}{|l|c|c|c|c|}
\hline
 & $\lambda_1$ & $\lambda_2$ & $\lambda_3$ & $\lambda_4$ \\  
\hline
Relative excess error (in $\%$) & 2 (0.3) & -20 (0.5) & -0.3 (0.3) & -5 (0.4)\\
\hline
\end{tabular}
\end{table}

These results suggest that the VIMP can identify the importance of the spatial coordinate $x$ when the covariates alone are insufficient to explain the intensity, as is the case for the models with $\lambda_2$ and $\lambda_4$. When the covariates fully explain the intensity, as in the models with $\lambda_1$ and $\lambda_3$, the importance of the spatial coordinates is much less pronounced, as expected.
Furthermore, Table~\ref{tab:error} shows that including the spatial coordinates does not deteriorate intensity estimation when they are unnecessary, while it substantially improves estimation accuracy when they are needed.

These findings suggest that, when in doubt, a cautious approach is to include the spatial coordinates in the intensity estimation procedure, alongside the covariates. The VIMP analysis can then help assess whether they are truly needed.

\subsection{Comparison with state-of-the-art methods}\label{xgboost}

In this section, we conduct an extensive simulation study to compare our estimator with competing state-of-the-art methods for nonparametric intensity estimation in the presence of covariates, under different scenarios. We follow the experimental setting of \cite{Lu06082025}, who introduced XGBoostPP and compared it with kernel smoothing \citep{guan2008,baddeley2012nonparametric}, the APP Bayesian approach of \citep{Kim22} and the deep neural network DKMPP of \cite{zhang2023} (see Section~\ref{sec:literature}). 
The scenarios considered  in  \cite{Lu06082025} covered intensity estimation in three settings: (1) two covariates, with estimation based on the two true covariates; (2) six covariates, with  estimation based on ten covariates (including four nuisance variables); and (3) two covariates, with estimation based on 10 covariates (including eight nuisance variables). In each case, the underlying point process was either a Poisson point process,  a log-Gaussian Cox process (LGCP), or a Neyman-Scott process (NSP),  with approximately 400 points on average. In their simulation study, XGBoostPP achieved the best results. Kernel smoothing and DKMPP performed poorly, as the former suffers from the curse of dimensionality, while the latter is not well suited to point patterns of moderate size due to its training requirements. Moreover, APP performed competitively in the simplest scenario with two covariates and no nuisance variables, but yielded poorer results in the other scenarios.

In our simulation study, we consider the same three scenarios as in \cite{Lu06082025} and  the same point process models, with the addition of a determinantal point process (DPP).  Recall that Poisson point processes exhibit no interaction between points, whereas LGCPs and NSPs exhibit attraction and DPPs  exhibit inhibition. For a given scenario, all considered point process models share the same intensity function, which depends on a parameter $\beta$ controlling the extent of spatial variation. The models differ in their interaction structure, governed by their respective interaction parameters. 
A detailed description of the experimental setting is provided in Appendix~\ref{appendix_simus}.
In each case, we compare our estimator with the two best-performing methods identified in \cite{Lu06082025} that are  APP and XGBoostPP. To select the hyperparameters of the latter method, we use the cross-validation procedure proposed in \cite{Lu06082025}, with the same grid of candidate values. For our method, we use the out-of-bag approach described in Section~\ref{sec:OOB}.

The results, based on 500 replications, are reported  in Tables~\ref{tab_poi}-\ref{tab_LGCP}, where our method is referred to as \texttt{spforest}. They show that APP outperforms the other methods in the simplest setting with two covariates and no nuisance variables (except for the DPP model). However, its performance deteriorates drastically in the presence of nuisance covariates.  In contrast, our method performs well across all scenarios and  models, and regularly outperforms the competing approaches. 

Concerning the comparison between our method and XGBoostPP, one limitation of our study is that both methods are sensitive to the choice of the hyperparameter grid, 
particularly for clustered point processes such as LGCP
and NSP models. In some cases, we had to resort to the rule of thumb described in Section~\ref{sec:OOB}. 
Consequently, we cannot rule out the possibility that alternative hyperparameter choices could affect the results, making it difficult  to conclude that one approach consistently outperforms  the other. Nevertheless,  our method  appears to perform robustly and accurately across all the settings considered.

\subsection{Computational complexity}\label{sec:comp}
From a computational standpoint, the training cost of a random forest grows linearly with the number of trees $M$, although the trees can be trained  in parallel to reduce the overall computation. In turn, the computational cost of each tree depends on the number of points in the point pattern, the number of covariates, and the tree depth, which is controlled by $minvol$ in  Algorithm~\ref{alg:random forest intensity}. As follows from this algorithm, each tree scales linearly with the number of points and covariates.

For the synthetic example in Figure~\ref{fig:non linear image}, training each tree took approximately 0.12 seconds on a standard  single-core  2.60GHz CPU. For the real-data example in Figure~\ref{fig:bei PP}, each tree took approximately 5 seconds, because they are much deeper. On the other hand, for the purely spatial random forest shown in Figure~\ref{fig:concrete} (on the plane) and Figure~\ref{fig:bei_mesh} (on a manifold), where the domain is represented by a fine mesh, the training time per tree is below  0.1 seconds.

\begin{table}[ht]
\centering
\caption{Mean integrated absolute errors (standard deviations) of
 intensity estimation for Poisson point process models.}
 \resizebox{\textwidth}{!}{
\begin{tabular}{l|cc|cc|cc}
Nb of covariates & \multicolumn{2}{c|}{2 (incl. 0 nuisance var.)} & \multicolumn{2}{c|}{10 (incl. 4 nuisance var.)} & \multicolumn{2}{c}{10 (incl. 8 nuisance var.)} \\
\hline
Parameter & $\beta=0.5$ & $\beta=1$ & $\beta=0.2$ & $\beta=0.4$ & $\beta=0.5$ & $\beta=1$ \\
\hline
APP         &59 (11) &74 (11)   &210 (15) &227 (8)  &192 (7)  &216 (7)  \\
XGBoost     &86 (16) &106 (14)  &108 (11) &147 (14) &111 (12) &138 (12)\\
spforest    &77 (11) & 105 (12) &88 (8)   &132 (8) &95 (9) & 124 (9) \\
\hline
\end{tabular}
}
\label{tab_poi}
\end{table}

\begin{table}[ht]
\centering
\caption{Mean integrated absolute errors (standard deviations) of
 intensity estimation for DPP models.}
 \resizebox{\textwidth}{!}{
\begin{tabular}{l|cc|cc|cc}
Nb of covariates & \multicolumn{2}{c|}{2 (incl. 0 nuisance var.)} & \multicolumn{2}{c|}{10 (incl. 4 nuisance var.)} & \multicolumn{2}{c}{10 (incl. 8 nuisance var.)} \\
\hline
Parameter & $\beta=0.2$ & $\beta=0.4$ & $\beta=0.1$ & $\beta=0.2$ & $\beta=0.2$ & $\beta=0.4$ \\
\hline
APP      &65 (10) &57 (10) &229 (37) &211 (14) &201 (23) &211 (14)  \\
XGBoost  &64 (13) &78 (16) &85 (13)  &108 (11) &87 (15)  &105 (12)  \\
spforest &47 (8)  &66 (9)  &57 (8)   &84 (7)   &57 (7)   &81 (8)    \\
\hline
\end{tabular}
}
\label{tab_dpp}
\end{table}

\begin{table}[ht]
\centering
\caption{Mean integrated absolute errors (standard deviations) of
 intensity estimation for LGCP models}
\resizebox{\textwidth}{!}{
\begin{tabular}{ll|llll|llll}
\multirow{2}{*}{\textbf{LGCP}} & \multirow{2}{*}{Covs}
& \multicolumn{2}{c|}{$\tau^2=1$}
& \multicolumn{2}{c|}{$\tau^2=2$}
& \multicolumn{2}{c|}{$\tau^2=1$}
& \multicolumn{2}{c}{$\tau^2=2$} \\
& & \multicolumn{1}{c}{$h = 2$}
& \multicolumn{1}{c|}{$h= 4$}
& \multicolumn{1}{c}{$h= 2$}
& \multicolumn{1}{c|}{$h= 4$}
& \multicolumn{1}{c}{$h= 2$}
& \multicolumn{1}{c|}{$h= 4$}
& \multicolumn{1}{c}{$h= 2$}
& \multicolumn{1}{c}{$h= 4$} \\ \hline
& & \multicolumn{4}{c|}{$\beta = 0.5$}
& \multicolumn{4}{c}{$\beta = 1$} \\
APP & \multicolumn{1}{l|}{2}
& 64 (14) & 75 (21)
& 74 (20) & 98 (38) 
& 81 (15) & 94 (27)  
& 94 (26) & 119 (50)  \\
XGBoostPP & \multicolumn{1}{l|}{(incl. 0}
& 91 (17) & 100 (21)
& 99 (20) & 117 (32) 
& 115 (15) &126 (23) 
& 127 (21) & 147 (38) \\
spforest &\multicolumn{1}{l|}{ nuisance)}
&93 (21) &104 (25)
&116 (26)&135 (37)  
&124 (14)&137 (21)
&144 (22)&162 (40) \\ \hline
& & \multicolumn{4}{c|}{$\beta = 0.2$}
& \multicolumn{4}{c}{$\beta = 0.4$} \\
APP & \multicolumn{1}{l|}{10}
&208 (10) &214 (13) 
&215 (14) &240 (31)  
&234 (9)  &245 (16)   
&246 (17) &272 (38)   \\
XGBoostPP & \multicolumn{1}{l|}{(incl. 4}
&107 (16) & 123 (29) 
&135 (35) & 180 (46) 
&149 (16) & 164 (25)  
&172 (26) & 203 (38) \\
spforest &\multicolumn{1}{l|}{ nuisance)}
&113 (12)&131 (18)
&131 (16)&163 (32)   
&147 (10)&162 (20) 
&165 (19)&194 (39) \\ \hline
& & \multicolumn{4}{c|}{$\beta = 0.5$}
& \multicolumn{4}{c}{$\beta = 1$} \\
APP & \multicolumn{1}{l|}{10}
&199 (9)   & 212 (15) 
&212 (15)  & 241 (35)  
&230 (11)  & 244 (21)  
&248 (23)  & 274 (44)  \\
XGBoostPP & \multicolumn{1}{l|}{(incl. 8}
&117 (20) & 136 (32)  
&153 (20) & 188 (43)  
&155 (20) & 174 (26) 
&187 (26) & 216 (37) \\
spforest &\multicolumn{1}{l|}{ nuisance)}
&120 (12)&136 (18)
&139 (18)&168 (33)   
&140 (13)&156 (22) 
&164 (23)&193 (43)\\ \cline{1-10}
\end{tabular}
}
\label{tab_LGCP}
\end{table}

\begin{table}[ht]
\centering
\caption{Mean integrated absolute errors (standard deviations) of
 intensity estimation for NSP models.}
\resizebox{\textwidth}{!}{
\begin{tabular}{ll|llll|llll}
\multirow{2}{*}{\textbf{NSP}} & \multirow{2}{*}{Covs}
& \multicolumn{2}{c|}{$\kappa=100$}
& \multicolumn{2}{c|}{$\kappa=200$}
& \multicolumn{2}{c|}{$\kappa=100$}
& \multicolumn{2}{c}{$\kappa=200$} \\
& & \multicolumn{1}{c}{$h = 2$}
& \multicolumn{1}{c|}{$h= 4$}
& \multicolumn{1}{c}{$h= 2$}
& \multicolumn{1}{c|}{$h= 4$}
& \multicolumn{1}{c}{$h= 2$}
& \multicolumn{1}{c|}{$h= 4$}
& \multicolumn{1}{c}{$h= 2$}
& \multicolumn{1}{c}{$h= 4$} \\ \hline
& & \multicolumn{4}{c|}{$\beta = 0.5$}
& \multicolumn{4}{c}{$\beta = 1$} \\
APP & \multicolumn{1}{l|}{2}
&76 (21) & 70 (19)
&70 (18) & 65 (16)
&101 (26) & 90 (22)
&87 (20) & 82 (17) \\
XGBoostPP & \multicolumn{1}{l|}{(incl. 0}
&102 (20) & 96 (20)
&96 (18) & 92 (17)
&138 (22) & 122 (20)
&120 (19) & 115 (17) \\
spforest & \multicolumn{1}{l|}{ nuisance)}
&108 (24) & 94 (21)
&113 (25) & 86 (17)
&146 (22) & 132 (17)
&132 (17) & 125 (14) \\ \hline
& & \multicolumn{4}{c|}{$\beta = 0.2$}
& \multicolumn{4}{c}{$\beta = 0.4$} \\
APP & \multicolumn{1}{l|}{10}
&220 (13) & 211 (10)
&210 (11) & 207 (11)
&256 (15) & 243 (13)
&240 (11) & 233 (9) \\
XGBoostPP & \multicolumn{1}{l|}{(incl. 4}
&157 (37) & 113 (18)
&113 (21) & 110 (13)
&190 (30) & 159 (19)
&158 (20) & 151 (16) \\
spforest & \multicolumn{1}{l|}{ nuisance)}
&142 (15) & 123 (16)
&123 (16) & 106 (14)
&175 (18) & 159 (16)
&156 (13) & 146 (11)  \\ \hline
& & \multicolumn{4}{c|}{$\beta = 0.5$}
& \multicolumn{4}{c}{$\beta = 1$} \\
APP & \multicolumn{1}{l|}{10}
&222 (15) & 208 (12)
&206 (11) & 200 (10)
&260 (18) & 230 (16)
&239 (13) & 228 (11)  \\
XGBoostPP & \multicolumn{1}{l|}{(incl. 8}
&173 (35) & 122 (23)
&125 (26) & 115 (17)
&200 (26) & 164 (23)
&167 (22) & 150 (17) \\
spforest & \multicolumn{1}{l|}{ nuisance)}
&149 (17) & 129 (16)
&129 (14) & 113 (15)
&177 (21) & 150 (17)
&151 (15) & 138 (14)  \\ \cline{1-10}
\end{tabular}
}
\label{tab_NSP}
\end{table}

\section{Theory} \label{sec:theory}

\subsection{Asymptotic framework and notation}\label{sec:asympt}

In a standard asymptotic framework, we let the number of points
tend to infinity. For a point process, this can be achieved in
several ways, the two most popular settings being the infill
asymptotic regime and the increasing domain asymptotic regime.  
Following \cite{Choiruddin21}, we consider a unifying framework
that contains the two previous settings, but also intermediate
ones.  Accordingly, we introduce a sequence of point processes $X_n$, $n\in\N$, assumed to be observed on $W_n\subset \R^d$ and with intensity $\lambda_n(.)=a_n \lambda(.)$, where 
$a_n>0$ is known. 
Our aim is to estimate $\lambda(x)$ at a given $x$.  Letting the mean number of observed points tend to infinity  implies  that $a_n|W_n|\to\infty$, whenever $\lambda$ is bounded.
The infill asymptotic framework is consistent with  $a_n\to
\infty$ and $W_n=W$ being constant, while the increasing domain regime
corresponds to $a_n=1$ and $|W_n|\to \infty$. In all cases, the
shape of the intensity of $X_n$ is $\lambda(.)$, making the
target of our nonparametric estimation problem meaningful. Note that the sequence $(a_n)$ can be viewed as a normalisation. It is introduced only for theoretical convenience: it encodes the densification of points in an infill asymptotic regime. However in practice, that is for a given fixed $n_0$,  there is no loss of generality in choosing $a_{n_0}=1$.

Beyond the number of points, the cross-correlation of $X_n$ may evolve
as $n\to\infty$. We denote by $g_n$ its pair correlation
function (see Appendix~\ref{sec:pcf}). In an increasing domain asymptotic framework, we will
typically have $g_n=g$ for all $n\in\N$, for some fixed
pair correlation function $g$, while in an infill asymptotic
framework, $g_n$ will change with $n$ and  the strength of
correlation will typically decrease with $n$, as exemplified
for several models in Appendix~\ref{sec:pcf}. Our general assumption, similar to [C6] in
\cite{Choiruddin21},  is that there exists $c>0$ such that for
any $A\subset W_n$ and any $n$,
\begin{equation}\label{hyp gn}
a_n\int_{A^2} |g_n(u,v)-1|\der u \der v \leq  c|A|.
\end{equation}
Lemma~\ref{lem gn}, stated in Appendix~\ref{sec:pcf}, shows that this condition typically boils down to $\sup_v \int_{\R^d} |g(u,v)-1|\der u<\infty$ for some underlying fixed pair correlation function $g$. The latter is a mild standard
assumption of weak dependence, as already considered in many
other studies, see for instance \cite{guan2008}. 
In particular, we show in Appendix~\ref{sec:pcf} that \eqref{hyp gn} is satisfied for standard models, such as inhomogeneous Poisson point processes, Neyman-Scott models, log-Gaussian Cox processes, Mat\'ern hardcore models and determinantal point processes, under mild assumptions.

The estimation of $\lambda(x)$ in this setting is carried out as explained in Section~\ref{sec:methodology}, where we add the subset $n$ in the notations to stress the dependence on $n$. Accordingly, assuming that $\lambda(.)=f(z(.))$ for some covariate $z:\R^d \to \R^p$, we consider $M$ partitions $\pi^{(1)}_n,\dots, \pi^{(M)}_n$ of $z(W_n)$. Assuming that $x$ is such that $z(x)\in z(W_n)$, we denote by $I^{(i)}_{n}(x)$ the cell in the partition $\pi^{(i)}_n$ that contains $z(x)$ and $$A_n^{(i)}(x)=z^{-1}(I_n^{(i)}(x))\cap W_n.$$
The tree intensity estimator of $\lambda(x)$ based on the partition 
$\pi^{(i)}_n$ is then
 \begin{equation}\label{tree estimator}
\hat\lambda^{(i)}(x) = \frac{|X_n\cap A_n^{(i)}(x)|}{a_n|A_n^{(i)}(x)|}
\end{equation}
and the random forest intensity estimator based on the $M$ partitions is given by
\begin{equation}\label{rf estimator}
\hat\lambda^{(RF)}(x) = \frac{1}{M} \sum_{i=1}^M \hat\lambda^{(i)}(x).
\end{equation}

For the theoretical analysis, we assume that each partition is generated independently of $X_n$, so that $\hat\lambda^{(RF)}(x)$ corresponds to a {\it  purely} random forest.

\subsection{Consistency}\label{sec: consistency}

We gather all needed assumptions for consistency below, including those discussed in the previous section.   For a
set $I$, we let  $\diam(I)=\sup_{y,z\in I}\|y-z\|$.

\begin{enumerate}[label={\bf(H\arabic*) },ref=H\arabic*,leftmargin=*,labelsep=0.5em]
\item The point process  $X_n$ is observed in $W_n\subset\R^d$, has intensity
$\lambda_n(.)=a_n\lambda(.)$ and its pair correlation function
$g_n$ satisfies \eqref{hyp gn}.  \label{H1}
\item $x\in \R^d$ is such that for all $n$,  
$z(x)\in z(W_n)$.  \label{hyp z(W_n)}
\item $\lambda(.)=f(z(.))$ where $z:\R^d\to\R^p$ and  where
$f:\R^p\to \R_+$ is a bounded function,  continuous at
$z(x)$.\label{hyp bounded}
\item The partitions $\pi^{(1)}_n,\dots, \pi^{(M)}_n$ of $z(W_n)$ are generated independently of $X_n$. \label{hyp pi}
\item For all $i=1,\dots,M$,  $\diam(I^{(i)}_{n}(x))\to 0$ in
probability as $n\to\infty$.\label{hyp I}
\item For all $i=1,\dots,M$, 
$\E\left(1/(a_n |A^{(i)}_n(x)|)\right)\to 0$ as $n\to\infty$. \label{hyp an}
\end{enumerate}

\begin{theo}\label{th consistency}
        Under  \eqref{H1}-\eqref{hyp an}, we have as $n\to
        \infty$, $$\E\left[\left(\hat\lambda^{(RF)}(x) -
        \lambda(x)\right)^2\right] \to 0.$$

\end{theo}

Hypothesis~\eqref{hyp I} demands that the partitions of $z(W_n)$
are such that the cell containing $z(x)$ concentrates around
$z(x)$.  This is a natural requirement for the bias to be
negligible. On the other hand, Hypothesis~\eqref{hyp an} and
Jensen's inequality imply (if $\lambda(x)\neq 0$) that the
expected number of observations of $X_n$ in $A^{(i)}_n(x)$ must
tend to infinity, which is also natural for the variance of
estimation to asymptotically vanish. Whether or not these
assumptions are met in practice is discussed in the following
examples: It depends on both the asymptotic regime (e.g. infill
or increasing domain) and the properties of the partitions $\pi^{(i)}_n$. 
 We refer to Section~\ref{sec:benefits covariates} for further theoretical insights, including rates of convergence in different settings. 
 
\medskip

\ex{ex no cov}
\noindent{\it Example \ref{ex no cov} (no covariate)}: Assume
that $z(u)=u$ for any $u\in \R^d$, that is the setting of Section~\ref{sec:purelyspatialpart}.  For the partitions $\pi_n^{(i)}$ of $z(W_n)=W_n$, consider  stationary tessellations with intensity $\gamma_n=h_n^{-d}$,
where $h_n>0$.  
With this notation, $h_n$ has the same interpretation
as the bandwidth in kernel estimation. 
Then the cell $I_n^{(i)}(x)=A_n^{(i)}(x)$ has the same distribution as
the zero cell of the tessellation $\pi_n^{(i)}$, up to vanishing edge effects, and we have $\E\left(1/|A_n^{(i)}(x)|\right)=O(h_n^{-d})$, see
\eqref{mean area} and Proposition~\ref{tess rates} in Appendix~\ref{sec:tessellations}. So 
\eqref{hyp an} is satisfied if  $a_n h_n^d\to \infty$. On the
other hand, \eqref{hyp I} is typically verified if $h_n\to 0$, as
for Poisson Vorono\"i,  Poisson Delaunay,  Poisson hyperplane and
STIT tessellations, see \cite{schneider2008,chiu2013} and
Appendix~\ref{sec:tessellations}. 
Consistency is thus ensured for these examples  whenever $h_n\to
0$ and $a_n h_n^d\to \infty$. Note that these conditions cannot
be met in an increasing domain asymptotic regime where $a_n=1$.
The lack of consistency in this setting is expected, since the number of events around $x$ does not increase and so the variance of estimation cannot vanish. 
In
other asymptotic regimes, consistency is ensured if  the total
number of points in the cell $I_n^{(i)}(x)=A_n^{(i)}(x)$, which is of order
$a_n h_n^d$,  tends to infinity while the diameter of the cell,
 of order $h_n$, tends to zero.

 \medskip

\ex{ex discrete}
 \noindent {\it Example \ref{ex discrete} (qualitative
 covariate)}: Assume that $z(.)$ is a binary variable, taking its
 values in $\{0,1\}$ (extension to more levels is
 straightforward) and that $z(W_n)=\{0,1\}$ for $n$ large enough,
 meaning that each level of $z$ is visited. 
 For the partitions $\pi_n^{(i)}$  of $z(W_n)$, it is natural to choose 
 the trivial deterministic partition $\{\{0\},
 \{1\}\}$. Then \eqref{hyp I} is obviously satisfied. 
  In turn, the set $A_n^{(i)}(x)$ is either
 $z^{-1}(0)\cap W_n$ or $z^{-1}(1)\cap W_n$ and  \eqref{hyp an}
 is satisfied if $a_n |A_n^{(i)}(x)|\to \infty$. 
 Therefore, if $a_n\to\infty$, as in the infill regime, consistency is ensured whenever $|A_n^{(i)}(x)|>0$,  or equivalently if the level sets of $z$ are not degenerated in $W_n$, in the sense that their volume is not zero. 
  In an increasing domain regime ($a_n=1$),
 consistency is ensured if $|A_n^{(i)}(x)|\to \infty$, meaning that
 each level set of $z$ covers an increasingly large  region of the
 observation domain. Note the advantage over Example~\ref{ex no
 cov}: By leveraging a (qualitative) covariate, consistency is
 possible even in an increasing domain asymptotic regime. This
 benefit has already been observed in \cite{guan2008} for a
 kernel estimator of the intensity based on a covariate, and is
 further investigated in the next section.\\

\ex{ex general}
\noindent {\it Example \ref{ex general} (general covariate)}: For
the partitions $\pi_n^{(i)}$ of $z(W_n)\subset\R^p$, consider
stationary tessellations  with  intensity $\gamma_n=h_n^{-p}$, as in
Example~\ref{ex no cov}, e.g., a Poisson Vorono\"i tessellation.
Then  \eqref{hyp I}  is satisfied whenever $h_n\to 0$.
Concerning \eqref{hyp an}, it is difficult to draw a general
statement, but the idea is that the number of events in the level
set $A_n^{(i)}(x)$ of $z$ must tend to infinity, even if the volume of
$A_n^{(i)}(x)$  typically tends to zero.
Let us present a
heuristic in an increasing domain asymptotic regime,  by
assuming that $z$ is the realisation of a stationary ergodic random
process.  
Then, for almost surely any realisation $z$, we anticipate by the ergodic theorem that 
$$|A_n^{(i)}(x)| = \int_{W_n} \1_{z(u)\in I_n^{(i)}(x)}\der u \approx |W_n|\, \P\left(Z\in   I_n^{(i)}(x) | I_n^{(i)}(x)\right),$$
where $Z$ follows the invariant distribution of the process $z$. The above probability is typically of order $| I_n^{(i)}(x)|$, so that $\E\left(|A_n^{(i)}(x)|\right)$ is of order $|W_n|\,\E(| I_n^{(i)}(x)|)$, which in turn is of order  $h_n^p |W_n|$. So we can expect  that \eqref{hyp an}  is satisfied if $a_n
h_n^p |W_n|\to \infty$. A formal treatment of this example is out
of the scope of this article. However, this heuristic confirms, as in the previous example, that the
introduction of a covariate makes it possible to ensure
consistency even when $a_n=1$.

\subsection{Rates of convergence and benefits of covariates}
\label{sec:benefits covariates}

As deduced from Section~\ref{sec: consistency}, in particular Example~\ref{ex
no cov}, a random forest built from partitions of $W_n$, that is without using the covariate $z$, is generally consistent to estimate the intensity $\lambda$ in an
infill asymptotic regime, even if  $\lambda$ actually depends on some
covariate $z$ through the relation $\lambda(u)=f(z(u))$.
In this section, we show the benefits of considering partitions of $z(W_n)$ instead of partitions of $W_n$, when the latter relation is
 trustworthy. 

A first advantage of using an estimator based on partitions of $z(W_n)$ is that it
allows for  the estimation of $\lambda(x)$ even for $x\notin W_n$,
provided $z(x)\in z(W_n)$ and $z(x)$ is known. This is useful
when it comes to predict the intensity outside the observation
region, where the covariate is observed but not the point process
of interest.
A second advantage is that the rate of convergence of
$\hat\lambda^{(RF)}(x)$  is generally improved when we consider
partitions of $z(W_n)$. As argued next, the global picture is as
follows:
\begin{enumerate}[label=(\roman*),leftmargin=2em,labelsep=0.3em]
\item In an increasing domain asymptotic regime ($a_n=1$ and
$|W_n|\to\infty$), estimation based on partitions of $W_n$ is
generally not consistent (see Example~ \ref{ex no cov}). In contrast, by 
leveraging a covariate $z$,  consistency can be achieved, 
provided that $z$ takes the same values  sufficiently often over $W_n$,
as is the case for a qualitative covariate  (Example~\ref{ex discrete}) or
the realisation of a stationary ergodic process  (Example~\ref{ex
general}).

\item In an infill asymptotic regime ($a_n\to\infty$ and
$W_n=W$), both approaches are generally consistent, and they can
both achieve the minimax rate of convergence when $\lambda$ and
$z$ are H\"older continuous. However in certain cases, as with a qualitative covariate, leveraging $z$ can lead to a strictly faster rate of convergence. 

\item In an intermediate asymptotic regime  ($a_n\to\infty$ and
$|W_n|\to\infty$), the estimator based on tessellations of
$z(W_n)$  generally converges faster than the one based on
$W_n$, provided that $z$ is sufficiently smooth and  takes repeated values frequently enough. 
\end{enumerate}

The first claim (i) is already clear from Examples~\ref{ex
discrete} and \ref{ex general}, see also \cite{guan2008}. To
support the two other claims, we first state the following rate
of convergence that involves a classical bias-variance tradeoff,
where the variance corresponds to the first term in the
right-hand side of \eqref{rate holder} below.  
To prove it, we strengthen 
Assumption~\eqref{hyp pi} by assuming that the partitions  $\pi_n^{(i)}$, $i=1,\dots,M$, are  independent and identically distributed. We then denote by $I_n(x)$ and $A_n(x)$ generic cells that have the same distributions as $I_n^{(i)}(x)$ and $A_n^{(i)}(x)$, respectively. 
Note that our
pointwise H\"older continuous assumption below implies a flat behaviour
of $f$ at $z(x)$ when $\beta>1$, the derivative being zero in
this case, and a standard H\"older regularity when $\beta\leq 1$.

\begin{prop}\label{prop rate holder}
In addition to \eqref{H1}-\eqref{hyp an}, assume that the
partitions $\pi_n^{(i)}$ are i.i.d. and that $f$ is pointwise
$\beta$-H\"older continuous at $z(x)$, for some $\beta>0$, i.e.,
$f(y)-f(z(x))=O(\|y-z(x)\|^\beta)$ as $y\to z(x)$.  Then the
purely random forest \eqref{rf estimator} satisfies,  for some
$c>0$ (depending on $x$ and $\beta$), 
\begin{equation}\label{rate holder}
\E\left[\left(\hat\lambda^{(RF)}(x) - \lambda(x)\right)^2\right] 
\leq 
c\, \E\left(\frac 1{a_n |A_n(x)|}\right) +  c\,\E\left(\diam(I_n(x))^{2\beta}\right) .
\end{equation}
\end{prop}

To appreciate the behaviour in an infill asymptotic regime, as
claimed in (ii), we introduce the simple deterministic
tessellation $\mathcal T_k(u)$, defined for $u\in\R^k$ as the
Vorono\"i tessellation in $\R^k$ with nuclei $h_n(\Z^k + u)$,
that is the simple lattice centered at $u$ with side length
$h_n$. The extension to {\it nice}  random tessellations, in the sense of Definition~\ref{nice} in Appendix~\ref{sec:tessellations}, is
straightforward for the first case (i.e., partitions of $W_n$)
but more technical for the second one (i.e., partitions of
$z(W_n)$), and we omit it.  Note that the
optimal rate obtained in both cases coincides with the minimax
rate of convergence established  in Theorem~6.5 in
\cite{kutoyants98} for  H\"older continuous
intensities. 

\begin{cor}\label{cor infill} In addition to the assumptions of
Proposition~\ref{prop rate holder}, assume  that $x\in int(W_n)$ and that $z$ is
$\alpha$-H\"older continuous at $x$, so that $\lambda$ is
$\alpha\beta$-H\"older continuous at $x$.
\begin{itemize}
\item If the $\pi_n^{(i)}$'s  are partitions of $W_n$, each being
equal to $W_n\cap \mathcal T_d(x)$ or being a nice tessellation with intensity $h_n^{-d}$, then, assuming  that $h_n\to 0$, for some $c>0$ and  $n$ large enough,
 $$\E\left[\left(\hat\lambda^{(RF)}(x) - \lambda(x)\right)^2\right]  
 \leq c \left(\frac 1 {a_n h_n^d} + h_n^{2\alpha\beta}\right).$$
\item  If the $\pi_n^{(i)}$'s  are partitions of $z(W_n)$, each
being equal to $z(W_n)\cap \mathcal T_p(z(x))$, then, assuming  that $h_n\to 0$, for some $c>0$ and  $n$ large enough,
$$\E\left[\left(\hat\lambda^{(RF)}(x) - \lambda(x)\right)^2\right]  
\leq c \left(\frac 1 {a_n h_n^{d/\alpha} }+ h_n^{2\beta}\right).$$
\end{itemize}
In both cases, the optimal rate  when $a_n\to\infty$  is
$a_n^{-2\alpha\beta/(2\alpha\beta+d)}$, achieved by choosing
$h_n=O(a_n^{-1/(2\alpha\beta+d)})$ in the first case and
$h_n=O(a_n^{-1/(2\beta+d/\alpha)})$  in the second case. 
\end{cor}

\begin{remark}
In regression settings, purely random forests do not always reach minimax rates. This is the case for the centred and median random forests in \cite{biau12} and \cite{klusowski21b}, and the balanced purely random forests considered in \cite{arlot2014}. In contrast, the forests based on STIT tessellations (including Mondrian forests) studied in  
\cite{mourtada2020, oreilly2021}, as well as the purely uniform random forests considered in \cite{Genuer2012,arlot2014}, can reach minimax rates. A key distinction between these approaches is that, in  the purely random forests that achieve minimax rates, the splitting rule for each cell depends on its geometric characteristics. This tends to produce regular cells, a property that appears to be crucial  for the optimality of random forests, as conjectured in \cite{arlot2014}. Proposition~\ref{prop rate holder} supports this intuition: when partitioning  $W_n$ so that $I_n(x)=A_n(x)$, an optimal cell should have a large volume $|A_n(x)|$ and a small diameter $\diam(I_n(x))$, as is the case for a ball or a regular polytope. As stated in Corollary~\ref{cor infill},  achieving this balance with nice tessellations yields minimax rates.
\end{remark}

\begin{remark}
It might be surprising at first glance that the dimension of the image space of $z$, that is $p$, does not appear in the second convergence rate based on partitions of  $z(W_n)$. To understand this, let us return to \eqref{rate holder}. The bias term depends on the diameter of the cell $I_n(x)$, a quantity that does not depend on the underlying dimension and is represented by $h_n$ in Corollary~\ref{cor infill}. The variance term,  in turn, involves the volume of the pre-image cell $A_n(x)$ of $I_n(x)$. While $I_n(x)$ lies in a $p$-dimensional space, $A_n(x)$ is $d$-dimensional, which explains why this term depends on $d$ rather than on $p$. Intuitively, as long as $z$ is smooth (H\"older continuous in the above results), points $y$ that are close to $x$ in $W_n$ are mapped to points $z(y)$ that are close to $z(x)$. Therefore, the pre-image $A_n(x)$ of a cell around $z(x)$ contains, in particular, points in a neighbourhood of $x$, so that its volume can be expected to be at least of the same order as that obtained when partitioning $W_n$ directly. 
 This is reflected in Corollary~\ref{cor infill}, which shows that, when $z$ is smooth and $h_n$ is appropriately chosen, partitions of $z(W_n)$ cannot perform worse than partitions of $W_n$. On the other hand, they can perform much better when the pre-image $A_n(x)$ is larger than a simple neighbourhood of $x$, as illustrated next.
\end{remark}

\begin{remark}
The assumption $x\in int(W_n)$ holds automatically true if $W_n$ expands in the direction of $x$ when $n\to\infty$.  On the other hand, if $x\in\partial W_n$, the same rates as above can be proven, provided $W_n$ has a smooth boundary. For instance if there exist $L\geq 1$ and $r_0>0$ such that $\forall y\in \partial W_n$, and $\forall r<r_0$, there exists $z$ satisfying $B(z,r/L)\subset B(y,r)\cap W_n$,    known as the corkscrew condition, then it is not difficult to extend Corollary~\ref{cor infill} to $x\in\partial W_n$. 
\end{remark}

The following corollary focuses on  a binary covariate $z$, and
shows that leveraging  $z$ improves the rate of convergence in
all asymptotic regimes, thus supporting the claims in (ii) and (iii) above in favour of using covariates. This setting corresponds to an extreme situation
of a smooth covariate that  takes repeated values frequently enough across space.  We recall that {\it nice} 
tessellations,  in the sense of Definition~\ref{nice} in
Appendix~\ref{sec:tessellations}, include stationary
Poisson Vorono\"i tessellations,  stationary Poisson Delaunay
tessellations, stationary Poisson hyperplane tessellations and 
STIT tessellations (see Proposition~\ref{prop scaling} in
Appendix~\ref{sec:tessellations}).
 \begin{cor}\label{cor discrete} In addition to the assumptions
of Proposition~\ref{prop rate holder}, assume  that $x\in int(W_n)$ and that $z$ is a
binary variable, i.e. $z(W_n)=\{0,1\}$, continuous at $x$, such
that
$|z^{-1}(z(x))\cap W_n|> c |W_n|$ where $c>0$. 
\begin{itemize}
\item If the $\pi_n^{(i)}$'s are partitions of $W_n$, built as
nice tessellations with intensity $h_n^{-d}$, then for some $c>0$,  provided
$h_n<a_n^{-\epsilon}$ for some $\epsilon> 0$, 
$$\E\left[\left(\hat\lambda^{(RF)}(x) - \lambda(x)\right)^2\right]  
\leq \frac c {a_n h_n^d}.$$
\item If the $\pi_n^{(i)}$'s correspond to the simple partition
$\{\{0\}, \{1\}\}$ of $z(W_n)$, then  for some $c>0$,
$$\E\left[\left(\hat\lambda^{(RF)}(x) - \lambda(x)\right)^2\right]  
\leq \frac c {a_n |W_n|}.$$
\end{itemize}
\end{cor}

To further support the claim (iii), we may consider the same
setting as in Corollary~\ref{cor infill} by assuming in addition
that $z$ is a periodic function. This is another instance of a
smooth covariate at $x$ that takes repeated values frequently enough. Then the bias for partitions of $z(W_n)$ is still
$h_n^{2\beta}$ while for the variance, we note that due to
periodicity $|A_n(x)|=O(|W_n| \times |z^{-1}(I_n(x))\cap W_1|)$
when $|W_n|\to \infty$. Moreover, by the same argument as in the
proof of Corollary~\ref{cor infill}, we may leverage the
$\alpha$-H\"older continuity of $z$ at $x$ to show that
$|z^{-1}(I_n(x))\cap W_1|=O(h_n^{d/\alpha})$ when $h_n\to 0$. We
then obtain that the rate of convergence in this case is of order 
$$\frac{1}{a_n|W_n| h_n^{d/\alpha}} + h_n^{2\beta},$$ provided
$|W_n|\to \infty$ and $h_n\to 0$, while the rate of convergence
based on partitions of $W_n$ remains similar as in
Corollary~\ref{cor infill}, leading to a slower rate when
$a_n\to\infty$ and $h_n$ is chosen as the optimal value. 

Alternatively, now suppose that in addition to the setting of
Corollary~\ref{cor infill}, $z$ is the realisation of a stationary ergodic
process in $\R^p$. Then following the heuristic in
Example~\ref{ex general}, the variance for partitions of $z(W_n)$
can be expected to be of order $1/(a_n h_n^p |W_n|)$ when
$|W_n|\to \infty$, while by  the $\alpha$-H\"older continuity of
$z$ it is also less than  $1/(a_n h_n^{d/\alpha})$  if $h_n\to
0$, see Corollary~\ref{cor infill}. The bias remains in turn of
order $h_n^{2\beta}$. This means that the optimal rate when both
$a_n\to\infty$ and $|W_n|\to\infty$ becomes $\min((a_n
|W_n|)^{-2\beta/(2\beta+p)},a_n^{-2\alpha\beta/(2\alpha\beta+d)})$.
This is to be compared with the optimal rate
$a_n^{-2\alpha\beta/(2\alpha\beta+d)}$ for partitions of $W_n$. In this case, the estimation
based on partitions of $z(W_n)$ cannot achieve a worse rate than
partitions based on $W_n$, and can be faster in some settings.

\subsection{Benefits of a random forest over a single tree}
\label{sec: benefit rf}

While it is clear from inequality \eqref{eqm tree} in the proofs that a purely random
forest performs at least as well as a single tree, the following
simple result helps understanding the possible gain offered by a
random forest. 

\begin{lem}\label{lem EQM}
If $\hat\lambda^{(RF)}(x)$, given by \eqref{rf estimator}, is a
purely random forest in the sense that the partitions are i.i.d.
and follow \eqref{hyp pi}, then
\begin{equation}\label{inequality RF}
\E\left[\left(\hat\lambda^{(RF)}(x) - \lambda(x)\right)^2\right] 
\leq 
\E\left[ \V(\hat\lambda^{(1)}(x)|\pi_n^{(1)})\right] + \frac 1 M  \V(B_n) + \E(B_n)^2,
\end{equation}
where $B_n=\E\left(\hat\lambda^{(1)}(x)|\pi_n^{(1)}\right) -
\lambda(x)$ is the conditional bias of a single tree. 
\end{lem}
As already observed in \cite{arlot2014,mourtada2020, oreilly2021}
for regression functions, since a single tree is a piecewise
constant function, its bias $B_n$ can be large when it comes to
estimate a smooth intensity function $\lambda(x)$.
While we can expect $\E(B_n)$ to alleviate this deficiency by the
averaging effect over the partitions' distribution, $\V(B_n)$
might be sub-optimal. A random forest, that averages a large
amount of single trees, becomes smoother than each of them, as reflected by the second term in \eqref{inequality RF} where $\V(B_n)$ is reduced
by a factor $M$.

Let us illustrate more specifically this phenomenon in the case
where there are no covariates, i.e. $z(u)=u$  as in
Example~\ref{ex no cov}, and in an infill asymptotic regime,
along similar lines as carried out in \cite{mourtada2020} and
\cite{oreilly2021} for regression functions.  By Corollary~\ref{cor
infill}, if $\lambda$ is a $\beta$-H\"older function
with $\beta\in (0,1]$, the  minimax rate of convergence of a
random forest is obtained whatever $M\geq 1$ and is thus also
achieved by a single tree ($M=1$). But if $\lambda$ is smoother
(but still not flat at $x$), more specifically if  its derivative
is ($\beta-1$)-H\"older continuous with $\beta\in (1,2]$, 
then the following proposition shows that  the rate of
convergence of a single tree is sub-optimal, while a random
forest can still achieve the minimax rate of convergence for $M$
large enough. Note that this assumption does not  imply the
pointwise $\beta$-H\"older continuity at $x$ with $\beta>1$, a
case where $\lambda$ is flat at $x$ and where a single tree
achieves the same rate as a random forest, as proved in
Corollary~\ref{cor infill}.
 
Denote  by $C^{1,\beta-1}$, for $\beta\in (1,2]$, the space of
functions $\lambda$ on $\widebar W  = \bigcup_n W_n$ that are
differentiable and satisfy $\sup_{u\in \widebar W} \|\nabla
\lambda(u)\|<\infty$ and for all $u,v\in \widebar W$, $\|\nabla
\lambda(u) - \nabla \lambda(v)\| \leq c\|u-v\|^{\beta-1}$ for
some $c>0$.

\begin{prop}\label{prop rate cov} 
Assume \eqref{H1}-\eqref{hyp
an} where $z(u)=u$. Suppose moreover that $x\in int(W_n)$ and that $\lambda\in C^{1,\beta-1}$. If the
partitions $\pi_n^{(i)}$ of $W_n$ are i.i.d., each built from a
nice tessellation with intensity $h_n^{-d}$, then for some $c>0$
\begin{equation}\label{rate stationary}
\E\left[\left(\hat\lambda^{(RF)}(x) - \lambda(x)\right)^2\right] 
\leq
c \left(\frac{1} {a_n h_n^d} + \frac{h_n^2} M + h_n^{2\beta}\right).
\end{equation}
When $a_n\to\infty$,  the minimax rate $a_n^{2\beta/(d+2\beta)}$
is obtained for $h_n=O(a_n^{-1/(d+2\beta)})$ and
$M>h_n^{2-2\beta}$. 
\end{prop}
As it appears clearly  in \eqref{rate stationary},  a single tree
($M=1$) achieves a sub-optimal rate of convergence in comparison
with a random forest having $M>h_n^{2-2\beta}$ trees. In fact, it
is not difficult to adapt Proposition~3 of  \cite{mourtada2020}
to our setting,  providing an example of intensity $\lambda\in
C^{1,\beta-1}$ for which  the sub-optimal upper bound \eqref{rate
stationary} when $M=1$ is also a lower bound for a single tree.
In contrast, the optimal rate $a_n^{2\beta/(d+2\beta)}$ obtained by a
random forest with $M>h_n^{2-2\beta}$ trees is minimax for
$\lambda\in C^{1,\beta-1}$  \cite[Theorem~6.5]{kutoyants98}.

\section{Proofs}
\label{sec:proof}
\subsection{Proof of Theorem~\ref{th consistency}}
By Jensen's inequality, 
\begin{align}
\E\left[\left(\hat\lambda^{(RF)}(x) - \lambda(x)\right)^2\right] 
&= 
\E\left[\left(\frac 1 M 
\sum_{i=1}^M (\hat\lambda^{(i)}(x) - \lambda(x))\right)^2\right] \nonumber \\
&\leq 
\frac 1 M  \sum_{i=1}^M  \E\left[\left(\hat\lambda^{(i)}(x) - \lambda(x)\right)^2\right], \label{jensen}
\end{align}
so that the mean square consistency of $\hat\lambda^{(RF)}(x)$
boils down to the consistency of each intensity tree estimator.
For $i=1$, we have by  the Pythagorean theorem
\begin{align}\label{biais-var}
        \E&\left[\left(\hat\lambda^{(1)}(x) - \lambda(x)\right)^2\right] \nonumber \\ 
        & = \E\left[\left(\hat\lambda^{(1)}(x) -  
                \E(\hat\lambda^{(1)}(x)|\pi_n^{(1)})\right)^2\right] +  
                \E\left[ \left( \E(\hat\lambda^{(1)}(x)|\pi_n^{(1)}) - 
                \lambda(x)\right)^2\right] \nonumber \\
        &= \E\left[ \V(\hat\lambda^{(1)}(x)|\pi_n^{(1)})\right]  + \E\left( B_n^2\right),
\end{align}
where $B_n=\E\left(\hat\lambda^{(1)}(x)|\pi_n^{(1)}\right) -
\lambda(x)$ is the conditional bias of the first tree. 

For the first term (the variance term), by definition of
$\lambda_n$ and $g_n$, using the fact that the partition $\pi_n^{(1)}$
is independent of $X_n$ by \eqref{hyp pi},
\begin{align*}
        \V(\hat\lambda^{(1)}(x)|\pi_n^{(1)})
        &=\frac{1}{a_n^2|A_n^{(1)}(x)|^2} 
                \E\left( \sum_{u\in X_n} \1_{u\in A_n^{(1)}(x)}^2 + 
                \sum_{u,v\in X_n}^{\neq} \1_{u,v\in A_n^{(1)}(x)}\bigg|\pi_n^{(1)} \right) \\
        & \hspace{0.5cm} -  \E^2\left(\hat\lambda^{(1)}(x)\bigg|\pi_n^{(1)} \right) 
        \\
        &=\frac{1}{a_n^2|A_n^{(1)}(x)|^2}  \int_{A_n^{(1)}(x)} \lambda_n(u)\der u \\ 
        &\hspace{0.5cm} + \frac{1}{a_n^2|A_n^{(1)}(x)|^2}  
        \int_{A^{(1)}_n(x)\times A^{(1)}_n(x)} \lambda_n(u)\lambda_n(v) (g_n(u,v)-1) \der u \der v.
\end{align*}
Since under \eqref{H1} and  \eqref{hyp bounded},
$\lambda_n=a_n\lambda$ where $\lambda$ is bounded, we obtain
using the property \eqref{hyp gn} assumed in \eqref{H1}
that $\V(\hat\lambda^{(1)}(x)|\pi_n^{(1)} )\leq c/(a_n|A^{(1)}_n(x)|)$ for
some $c>0$. Hence the  first term in \eqref{biais-var} tends to 0
by \eqref{hyp an}.

For the second term (the bias term), we have by definition of
$\lambda_n$,
\begin{align}\label{form Bn}
        B_n&= \frac{1}{a_n|A_n^{(1)}(x)|}\int_{A_n^{(1)}(x)} \lambda_n(u)\der u - \lambda(x) 
        =  \frac{1}{|A_n^{(1)}(x)|}\int_{A_n^{(1)}(x)}  ( f(z(u))-f(z(x))) \der u.
\end{align}
Let $\epsilon>0$, then
$$\P(\forall u\in A_n^{(1)}(x),  \left | f(z(u))-f(z(x))\right|<\epsilon)\leq \P(|B_n|<\epsilon)$$
and since $u\in A_n^{(1)}(x) \Leftrightarrow z(u)\in I_n^{(1)}(x)$, this means that 
$$\P(\forall y\in I_n^{(1)}(x),  \left |
f(y)-f(z(x))\right|<\epsilon)\leq \P(|B_n|<\epsilon).$$ By
continuity of $f$ at $z(x)$, as assumed in  \eqref{hyp
bounded}, and since $z(x)\in I_n^{(1)}(x)$, there exists $\eta>0$ such that $\diam(I_n^{(1)}(x))<\eta$
implies  $\left | f(y)-f(z(x))\right|<\epsilon$ for all $y\in
I_n^{(1)}(x)$. Hence 
$$\P(\diam(I_n^{(1)}(x))<\eta)\leq  \P(|B_n|<\epsilon),$$ 
whereby $B_n$
tends to 0 in probability by \eqref{hyp I}. Since $B_n$ is
uniformly bounded thanks to \eqref{hyp bounded}, the sequence
$(B_n^2)$ is uniformly integrable and we deduce that
$\E(B_n^2)\to 0$ showing that the second term in
\eqref{biais-var} tends to 0.

\subsection{Proof of Proposition~\ref{prop rate holder}}

Since the partitions are i.i.d., we deduce  from \eqref{jensen} and \eqref{biais-var}
  that
\begin{equation}\label{eqm tree}
\E\left[\left(\hat\lambda^{(RF)}(x) - \lambda(x)\right)^2\right] 
\leq 
\E\left[ \V(\hat\lambda^{(1)}(x)|\pi_n^{(1)})\right] + \E(B_n^2),
\end{equation}
where $B_n=\E\left(\hat\lambda^{(1)}(x)|\pi_n^{(1)}\right) -
\lambda(x)$. From the proof of Theorem~\ref{th consistency}, we deduce that $ \E\left[
\V(\hat\lambda^{(1)}(x)|\pi_n^{(1)})\right]\leq c\,
\E\left(1/(a_n |A_n(x)|)\right)$. On the other hand, we have
from \eqref{form Bn}
$$\E(B_n^2)= \E\left[ \left(   \frac{1}{|A_n(x)|}\int_{A_n(x)}
(f(z(u))-f(z(x)))\der u \right)^2\right].$$ By the pointwise
$\beta$-H\"older continuity assumption of $f$ at $z(x)$, there
exists a vicinity $V_x$ of $z(x)$ and $L>0$ such that for
all $y\in V_x$,  $|f(y)-f(z(x))|\leq L  \|y-z(x)\|^\beta$. Since
$f$ is bounded, we have for some $c>0$, 
\begin{align*}
        &\frac{1}{|A_n(x)|}\int_{A_n(x)} |f(z(u))-f(z(x))| \der u \\
        &\leq  \frac{1}{|A_n(x)|} \int_{A_n(x)}  
               \1_{z(u)\in V_x}  |f(z(u))-f(z(x))|\der u + 
               \frac{c}{|A_n(x)|}\int_{A_n(x)}  \1_{z(u)\notin V_x}\der u \\
        &\leq \frac{L}{|A_n(x)|} \int_{A_n(x)} \|z(u)-z(x)\|^\beta\der u +  
                \frac{c}{|A_n(x)|}\int_{A_n(x)}  \1_{z(u)\notin V_x}\der u.
\end{align*}
Denote by $\eta_x>0$ the radius of a ball centred at $z(x)$ and
included in $V_x$. If $u\in A_n(x)$, meaning that $z(u)\in
I_n(x)$, and $z(u)\notin V_x$, then $\diam(I_n(x))>\eta_x$.
Therefore
\begin{align*}
\frac{1}{|A_n(x)|}\int_{A_n(x)}  \1_{z(u)\notin V_x} \der u 
&\leq  
\1_{\diam(I_n(x))>\eta_x}.
\end{align*}
By Markov's inequality, we obtain that for some $c>0$ depending on $x$ and $\beta$, 
$$\E(B_n^2)\leq c \, \E(\diam(I_n(x))^{2\beta}).$$

\subsection{Proof of Corollary~\ref{cor infill}} 

The rate of
convergence in the first case  (i.e., partitions of $W_n$) is
given by Proposition~\ref{prop rate holder} where $f=\lambda$ is
$\alpha\beta$-H\"older continuous at $x$. 
If the partitions
are nice tessellations with intensity $h_n^{-d}$, in the sense of Definition~\ref{nice}, then $A_n(x)=I_n(x)=Z_n(x)\cap W_n$ where  $Z_n(x)$ corresponds to the associated zero-cell given $x$.  By Proposition~\ref{tess rates} in Appendix~\ref{sec:tessellations},  $\E(\diam(I_n(x))^{2\alpha\beta})\leq c h_n^{2\alpha\beta}$, for some $c>0$. The same inequality holds trivially true for the deterministic tessellation $\mathcal T_d(x)\cap W_n$. On the other hand, for a nice tessellation, still by Proposition~\ref{tess rates}, $\E\left(1/(a_n |A_n(x)|)\right)\leq c/(a_n h_n^{d})$ whenever  $h_n$ is small enough. The last inequality also remains true for  the deterministic tessellation $\mathcal T_d(x)\cap W_n$ since $|A_n(x)|=|T_d(x)\cap W_n|\geq |B(x,h_n/2)\cap W_n|$. This yields the result.




In the second case (i.e., partitions of $z(W_n)$), we deduce from
Proposition~\ref{prop rate holder} that the bias is of order
$\E\left(\diam(I_n(x))^{2\beta}\right)=O(h_n^{2\beta})$ for the
deterministic tessellation  $\mathcal T_p(z(x))\cap z(W_n)$. For
the variance term, note that since $z$ is $\alpha$-H\"older at
$x$, there exists a vicinity $V_x$ of $x$ in $W_n$ such that $y\in V_x$
implies $\|z(x)-z(y)\|\leq c_x \|x-y\|^\alpha$ for some $c_x>0$.
We deduce that if  $y\in W_n$ is such that $\|x-y\|<(h_n/(2c_x))^{1/\alpha}$, and if $h_n$
is small enough, then $y\in V_x$ and  $\|z(x)-z(y)\|\leq h_n/2$.
Since  by definition $I_n(x)$ is simply the intersection of $z(W_n)$ with 
the cube centred at $z(x)$ with side length $h_n$, the latter
implies that $z(y)$ belongs to $I_n(x)$. Hence, for $h_n$ small
enough
$$y\in W_n \text{ and }\|x-y\|<\left(\frac{h_n}{2c_x}\right)^{1/\alpha} \Longrightarrow z(y)\in I_n(x).$$
We deduce that 
\begin{align*}
        |A_n(x)|=\int_{W_n} \1_{z(y)\in I_n(x)} \der y  
        & \geq \int_{W_n} \1_{\|x-y\|<(h_n/(2c_x))^{1/\alpha}} \der y\\ & = \left|W_n \cap B\left(x,(h_n/(2c_x))^{1/\alpha}\right)\right|\geq c  h_n^{d/\alpha},
\end{align*}
for some $c>0$, where the last inequality follows from the fact that $x\in int(W_n)$.  This entails the rate of convergence in the second case.


The optimal rate is obtained when the bias term and the
variance term are of the same order,  which is  achieved by
choosing $h_n=O(a_n^{-1/(2\alpha\beta+d)})$ in the first case and
$h_n=O(a_n^{-1/(2\beta+d/\alpha)})$  in the second case.

\subsection{Proof of Corollary~\ref{cor discrete}}

For partitions of $W_n$, $A_n(x)=I_n(x)$. Moreover, by continuity
of $z$ at $x$ and of $f$ at $z(x)$, and since $z$ is a binary
variable, then $\lambda$ is constant in a vicinity of $x$. For
this reason, for any $\beta>0$,  $\lambda$ is pointwise
$\beta$-H\"older at $x$. This means that \eqref{rate holder} in
Proposition~\ref{prop rate holder} holds true in this case
for $A_n(x)=I_n(x)$ and for any $\beta>0$. 
If the partitions
are nice tessellations with intensity $h_n^{-d}$, in the sense of Definition~\ref{nice}, then $A_n(x)=Z_n(x)\cap W_n$ where  $Z_n(x)$ corresponds to the associated zero-cell given $x$.  By Proposition~\ref{tess rates} in Appendix~\ref{sec:tessellations}, $\E\left(1/(a_n |A_n(x)|)\right)\leq c/(a_n h_n^{d})$ and $\E(\diam(A_n(x))^{2\beta})\leq c h_n^{2\beta}$. Since by assumption
$h_n<a_n^{-\epsilon}$ for some $\epsilon> 0$, the choice
$2\beta>1/\epsilon-d$ leads to $h_n^{2\beta}<1/(a_n h_n^{d})$ and
the result of Corollary~\ref{cor discrete} follows.


For partitions of $z(W_n)$ corresponding to $\{\{0\}, \{1\}\}$,
we have $A_n(x)=z^{-1}(z(x))\cap W_n$ and so $z(u)=z(x)$ for all
$u\in A_n(x)$. This entails $B_n=0$ (using the same notation as
in the proof of Proposition~\ref{prop rate holder}) and the mean
square error has the same order as $\E\left(1/(a_n
|A_n(x)|)\right)$, which by assumption is of order $1/(a_n
|W_n|)$.

\subsection{Proof of Lemma~\ref{lem EQM}}
By  the Pythagorean theorem,
\begin{multline*}
\E\left[\left(\hat\lambda^{(RF)}(x) - \lambda(x)\right)^2\right] 
 = \E\left[\left(\hat\lambda^{(RF)}(x) -  \E(\hat\lambda^{(RF)}(x)| 
 \pi_n^{(1)},\dots,\pi_n^{(M)})\right)^2\right] 
 \\ +  
 \E\left[ \left( \E(\hat\lambda^{(RF)}(x)|\pi_n^{(1)},\dots,\pi_n^{(M)}) - 
 \lambda(x)\right)^2\right].
\end{multline*}
On the one hand, by Jensen's inequality, and since the
partitions are i.i.d.,
\begin{align*}
        \E&\left[\left(\hat\lambda^{(RF)}(x) -  \E(\hat\lambda^{(RF)}(x)|
        \pi_n^{(1)},\dots,\pi_n^{(M)})\right)^2\right]
        \\
        &\leq \frac 1 M \sum_{i=1}^M \E\left[
                \left(\hat\lambda^{(i)}(x) -  
                \E(\hat\lambda^{(i)}(x)|\pi_n^{(1)},\dots,\pi_n^{(M)})\right)^2
                \right]
                \\
                & = \frac 1 M \sum_{i=1}^M \E\left[\left(\hat\lambda^{(i)}(x) -  
                \E(\hat\lambda^{(i)}(x)|\pi_n^{(i)})\right)^2\right]
                = \E\left[ \V(\hat\lambda^{(1)}(x)|\pi_n^{(1)})\right].
\end{align*}
On the other hand, again by the i.i.d. property of the partitions,
\begin{align*}
        \E&\left[ \left( \E(\hat\lambda^{(RF)}(x)|
        \pi_n^{(1)},\dots,\pi_n^{(M)}) - \lambda(x)\right)^2\right]
        \\
        &= \V\left[ \frac 1 M \sum_{i=1}^M 
        \E\left(\hat\lambda^{(i)}(x)|\pi_n^{(1)},\dots,\pi_n^{(M)}\right)\right] + 
        \left[\E(\hat\lambda^{(RF)}(x)) - \lambda(x)\right]^2
        \\
        &= \frac 1 M  \V\left(\E\left(\hat\lambda^{(1)}(x)|\pi_n^{(1)}\right)\right) + 
        \left[\E(\hat\lambda^{(1)}(x)) - \lambda(x)\right]^2
                = \frac 1 M  \V(B_n) + \E(B_n)^2.
\end{align*}

\subsection{Proof of Proposition~\ref{prop rate cov}}

We start from the upper bound in \eqref{inequality RF}
obtained in Lemma~\ref{lem EQM}.  For the first term, we deduce from the proof of
Theorem~\ref{th consistency} that 
$$\E\left[ \V(\hat\lambda^{(1)}(x)|\pi_n^{(1)})\right] \leq c\,
\E\left(1/(a_n |A_n(x)|)\right).$$ 
For the second term  in
\eqref{inequality RF}, observe that $\lambda$ is a Lipschitz
function so that 
\begin{align*}
        \V(B_n)\leq \E(B_n^2) & = \E\left[ \left( 
        \frac{1}{|A_n(x)|}\int_{A_n(x)} (\lambda(u)-\lambda(x))\der u\right)^2\right]
        \\
        &\leq c\, \E\left[ \left(   \frac{1}{|A_n(x)|}\int_{A_n(x)} \|u-x\|\der u\right)^2\right]
        \\
        &\leq c\, \E\left(\diam(A_n(x))^2\right).
\end{align*}
where 
$c>0$. For the third term, denoting $F_n(u)=\E(\1_{A_n(x)}(u)/|A_n(x)|)$, we have
$$\E(B_n) = \int (\lambda(u)-\lambda(x) - \nabla
 \lambda(x)'(u-x)) F_n(u)\der u +  \int  \nabla \lambda(x)'(u-x)
 F_n(u)\der u,$$ 
 where for $v\in \R^d$, $v'$ stands for the
 transpose of $v$. Since $\lambda\in C^{1,\beta-1}$, we obtain
 by a Taylor expansion that $\|\lambda(u)-\lambda(x) - \nabla
 \lambda(x)'(u-x)\|\leq c\|u-x\|^{\beta}$ for some $c>0$, so that
\begin{align*}
\E(B_n)^2&\leq 2c \left(\int \|u-x\|^{\beta} F_n(u)\der u\right)^2 + 
2 \| \nabla \lambda(x)\|^2 \left\| \int (u-x) F_n(u)\der u\right\|^2
\\
&\leq 2c\,  \E\left(\diam(A_n(x))^\beta\right)^2 + 2 c \left\| 
        \int (u-x) F_n(u)\der u\right\|^2.
\end{align*}

For a stationary tessellation, $\int (u-x) F_n(u)\der u=0$, see the
argument of \cite[Lemma~16]{oreilly2021} that is valid for any
stationary tessellation. If moreover this tessellation is nice with intensity
$h_n^{-d}$, we have  by Proposition~\ref{tess rates}  $\E(1/(a_n |A_n(x)|)\leq c/(a_n h_n^d)$ and 
 $\E\left(\diam(A_n(x))^\beta\right) \leq c\,
h_n^{\beta}$ for some $c>0$, leading to the result.

\begin{appendix}
\section{Some additional details on random tessellations}\label{sec:tessellations}

  A tessellation of $\R^d$ is a partition of $\R^d$ into non-empty
 compact and convex polytopes. Denoting by $\mathcal K$ the set of
 such polytopes, a tessellation can be viewed as a collection of
 cells belonging to $\mathcal K$.  For basic materials concerning
 random tessellations, we refer the reader to \cite{chiu2013} and
 \cite{schneider2008}. We consider in this section stationary
 tessellations in $\R^d$, see \cite{moller1989} for an overview
 and a study of their basic characteristics. Among them, the
 intensity, the typical cell and the zero cell are of primary
 importance. The intensity  $\gamma$  represents the mean number
 of cells per unit measure. The typical cell $Z_\gamma$ can be
 viewed as a randomly chosen cell among all cells of the
 tessellation. The zero cell $Z_\gamma(0)$ is simply the cell  that
 contains the origin. Note that by stationarity the law of the zero
 cell is the same as the law of the cell containing any given point
 $x\in\R^d$. A formal definition of these characteristics can be
 found in the above references.
 
 For a stationary tessellation with intensity $\gamma$, we have by  \cite[Corollary 5.2]{moller1989}:
 \begin{equation}\label{mean area}
 \E\left(\frac 1 {|Z_\gamma(0)|}\right) = \frac 1 {\E(|Z_\gamma|)} = \gamma.
 \end{equation}
 
More generally, see for instance (5.2)
 in  \cite{moller1989}: for any measurable function $\varphi:
 \mathcal K\mapsto\R_+$ which is translation-invariant, that is $\varphi(Z+x)=\varphi(Z)$ for any $Z\in\mathcal K$ and $x\in\R^d$, then 
\begin{equation}\label{zero palm}
\E\left(\varphi(Z_\gamma(0))\right)=\gamma\E(\varphi(Z_\gamma) |Z_\gamma|).
\end{equation}

 \begin{definition}[scaling property]
 A stationary random tessellation in $\R^d$ with intensity
 $\gamma$ has the scaling property if its typical cell $Z_\gamma$
 satisfies the equality in distribution $$Z_\gamma \overset{d}{=}
 \gamma^{-1/d} Z_1.$$
 \end{definition}

 As an immediate consequence of the scaling property, if $\varphi:
 \mathcal K\mapsto\R_+$ is a measurable $\alpha$-homogeneous
 function for some $\alpha>0$, i.e.,
 $\varphi(aC)=a^\alpha\varphi(C)$  for any $a>0$ and $C\in\mathcal
 K$, then $\varphi(Z_\gamma)\overset{d}{=} \gamma^{-\alpha/d}
 \varphi(Z_1)$. We then obtain the following result. 
 \begin{prop}\label{prop scaling general} Let $Z_\gamma$ and
 $Z_\gamma(0)$ be the typical cell and the zero cell,
 respectively, of a stationary random tessellation in $\R^d$ with
 intensity $\gamma$ having the scaling property.  Let $\varphi:
 \mathcal K\mapsto\R_+$ be a measurable $\alpha$-homogeneous
 function. Then $\E(\varphi(Z_\gamma))= \gamma^{-\alpha/d}
 \E(\varphi(Z_1))$. If moreover $\varphi$ is invariant by
 translation, then $\E(\varphi(Z_\gamma(0))= \gamma^{-\alpha/d}
 \E(\varphi(Z_1(0))$.
 \end{prop}
 
 \begin{proof}
 The first relation is an immediate consequence of the scaling property.  The second one is a consequence
 of \eqref{zero palm}.
 \end{proof}

 \begin{definition}[Nice tessellation]\label{nice}
 A tessellation with intensity $\gamma$ is nice if it is stationary, has the
 scaling property and for any $t>0$,
 $$\P(\diam(Z_\gamma)>t)< a \exp(-b \gamma^{\tau/d} t^\tau),$$ 
 for some positive constants $a,b,\tau$ that do not depend on $\gamma$. 
  \end{definition}

   \begin{prop}\label{prop scaling}  Stationary Poisson Vorono\"i
 tessellations, stationary Poisson Delaunay tessellations, 
 stationary Poisson hyperplane tessellations and STIT
 tessellations (including the Mondrian process as a particular
 case) are nice tessellations in $\R^d$. 
 \end{prop}

   \begin{proof}
 For a stationary Poisson Delaunay tessellation, the distribution
 of $Z_\gamma$ is given in \cite[Theorems~10.4.4]{schneider2008}
 from which we easily deduce the scaling property. For a STIT
 process, this is \cite[Lemma~5]{nagel2005}. Moreover, the
 distribution of the typical cell of a STIT process is similar as
 the typical cell of a stationary Poisson hyperplane tessellation
 with the same characteristics, see \cite[Lemma~3]{nagel2003} and
 \cite[Corollary~1]{Schreiber13}. The scaling property of the
 stationary Poisson hyperplane tessellation thus follows, see also
 \cite[Theorem~10.4.6]{schneider2008} for an explicit expression
 of the distribution of $Z_\gamma$ in the isotropic case. For a
 stationary Poisson Vorono\"i tessellation, $Z_\gamma\overset{d}{=}
 C(0|X_\gamma\cup\{0\})$, that is the Vorono\"i cell with nucleus 0
 in $X_\gamma\cup\{0\}$ where $X_\gamma$ denotes the Poisson point
 process with intensity $\gamma$, 
 see  \cite{mollerLNS94}. Since $X_\gamma  \overset{d}{=}
 \gamma^{-1/d} X_1$,  $C(0|X_\gamma\cup\{0\}) \overset{d}{=}
 C(0 | \gamma^{-1/d} X_1 \cup\{0\})$, the latter cell being exactly
 $\gamma^{-1/d} C(0 | X_1\cup\{0\})$ by definition of a Vorono\"i
 cell. So  $Z_\gamma \overset{d}{=}  \gamma^{-1/d} Z_1$. 
 
 Regarding the exponential tail of $\diam(Z_\gamma)$, the result for any $\gamma>0$ is deduced from the case $\gamma=1$ by the scaling property.  
For the Delaunay and the Vorono\"i tessellation, the exponential tail with $\tau=d$ is proven by noting that $\P(\diam(Z_1)>t)$ boils down to the void probability, under the Poisson point process, of a region with volume proportional to $t^d$. We refer to \cite[Lemma~1]{FossZuyev} for details in the planar case, that extend straightforwardly to any dimension.  For the hyperplane tessellation, the tail with $\tau=1$ is proven in \cite[Theorem~2]{hug2007} for the zero cell. We deduce the same result for the typical cell $Z_\gamma$ by the coupling $Z'_\gamma\subset Z_\gamma(0)$ almost surely,  where $Z'_\gamma$ has the same distribution as $Z_\gamma$, see \cite[Corollary 10.4.1]{schneider2008}. This finally also yields the exponential tail for the STIT tessellation, since its typical cell is the same in distribution as for a hyperplane tessellation.
  \end{proof}

 \begin{prop}\label{tess rates}
 Let $W\subset \R^d$ so that $0\in int(W)$.  For a nice tessellation  with intensity $\gamma$, we have, for any $\beta\geq 0$,
 \begin{align}
 &\E\left(\diam(Z_\gamma(0)\cap W)^\beta\right) \leq c_1\gamma^{-\beta/d},\label{scaling diam}\\
 &\E\left(\frac{1}{|Z_\gamma(0)\cap W|}\right)\leq  \gamma + c_2 \exp(-c_3 \gamma^{\tau/d}), \label{scaling area}
 \end{align}
 where $0<c_1,c_2,c_3,\tau<\infty$ do not depend on $\gamma$.
 \end{prop}

\begin{proof}
Since $Z_\gamma(0)\cap W \subset Z_\gamma(0)$, it suffices to prove \eqref{scaling diam}  for $W=\R^d$. First note that by \eqref{zero palm}, $$\E(\diam(Z_1(0))^k)=\E(\diam(Z_1)^k |Z_1|)\leq c\, \E(\diam(Z_1)^{k+d})<\infty,$$
where $c>0$ and the finiteness is due to the exponential tail of $\diam(Z_1)$, satisfied by a nice tessellation.
 The result  is then an immediate consequence of Proposition~\ref{prop scaling
 general} with $\varphi=\diam^k$ and $\alpha=k$, letting 
 $c_1:=\E(\diam(Z_1(0))^k)$.
 
To prove  \eqref{scaling area}, it is sufficient to show the inequality for $B(0,\rho)$ instead of $W$, where $\rho>0$ is such that $B(0,\rho)\subset W$, its existence being enforced by the assumption $0\in int(W)$. If $\rho\geq \diam(Z_\gamma(0))$, then $Z_\gamma(0)\subset B(0,\rho)$ and 
 $|Z_\gamma(0)\cap B(0,\rho)|=|Z_\gamma(0)|$. If now $\rho<\diam(Z_\gamma(0))$, by convexity of $Z_\gamma(0)$, $$\frac{\rho}{\diam(Z_\gamma(0))} Z_\gamma(0) \subset Z_\gamma(0).$$
Moreover, for any $x\in Z_\gamma(0)$, $\rho |x|/ \diam(Z_\gamma(0)) \leq \rho$, so that
$$\frac{\rho}{\diam(Z_\gamma(0))} Z_\gamma(0) \subset Z_\gamma(0)\cap B(0,\rho).$$
We deduce that
$$\frac{1}{|Z_\gamma(0)\cap B(0,\rho)|} \leq \frac 1 {|Z_\gamma(0)|}\1_{\rho\geq \diam(Z_\gamma(0))} + \frac{\diam(Z_\gamma(0))^d}{\rho^d} \frac 1 {|Z_\gamma(0)|} \1_{\rho< \diam(Z_\gamma(0))}.$$
Taking the expectation, we have 
$$\E\left(\frac{1}{|Z_\gamma(0)\cap B(0,\rho)|}\right)\leq  \E\left(\frac{1}{|Z_\gamma(0)|}\right) + \frac 1 {\rho^d} \E\left(\frac{\diam(Z_\gamma(0))^d}{|Z_\gamma(0)|} \1_{\rho< \diam(Z_\gamma(0))}\right).$$
Using  \eqref{mean area} for the first term, \eqref{zero palm} for the second term and Cauchy-Scwharz inequality, 
\begin{align*}
\E\left(\frac{1}{|Z_\gamma(0)\cap B(0,\rho)|}\right)& \leq \gamma +  \frac{\gamma}{\rho^d}  \E\left(\diam(Z_\gamma)^d \1_{\rho< \diam(Z_\gamma)}\right)\\
&\leq \gamma + \frac{\gamma}{\rho^d}  \E\left(\diam(Z_\gamma)^{2d}\right)^{1/2} \P\left(\diam(Z_\gamma)>\rho\right)^{1/2}.
\end{align*}
By the scaling property  $\E\left(\diam(Z_\gamma)^{2d}\right)=c\gamma^{-2}$ where $c=\E(\diam(Z_1)^{2d})$ is finite since $\diam(Z_1)$ has an exponential tail. We get
\begin{align*}
\E\left(\frac{1}{|Z_\gamma(0)\cap B(0,\rho)|}\right)\leq  \gamma +  \frac{\sqrt c}{\rho^d} a \exp\left(-\frac b 2  \gamma^{\tau/d} \rho^\tau\right),
\end{align*}
concluding the proof.
  \end{proof}

  \section{Pair correlation function and asymptotic regimes}\label{sec:pcf}
  
While $\lambda$ encodes the first moment of a spatial point process, the pair correlation function encodes its second order properties, see  \cite{baddeley15book,Coeurjolly2019understanding}. Let us first recall that  the second order intensity $\lambda^{(2)}$  of the process, when it exists, is the function that  satisfies 
for any Borel sets $B_1,B_2 \subset \R^d$,
\begin{equation*}
\E \sum_{u,v\in X}^\neq \cara{u \in B_1, v\in B_2} = \int_{B_1\times B_2} \lambda^{(2)}(u,v) \der u\der v.
\end{equation*}
If there is no interaction, as for a Poisson point process, $\lambda^{(2)}(u,v)=\lambda(u)\lambda(v)$. In turn, the pair correlation function (pcf) is defined for any $u,v\in \R^d$ by 
$$g(u,v)=\frac{\lambda^{(2)}(u,v)}{\lambda(u)\lambda(v)},$$
provided  $\lambda(u)\lambda(v)\neq 0$, otherwise $g(u,v)=0$. 

Coming back to the setting of Section~\ref{sec:asympt}, we consider a sequence of point processes  $X_n$, each with intensity $\lambda_n=a_n\lambda$ and pair correlation $g_n$. Lemma~\ref{lem gn} below provides useful conditions under which assumption  \eqref{hyp gn} is satisfied. We show in the following that they are satisfied for a wide class of spatial point process models. 

  \begin{lem}\label{lem gn} Let $g$ be a pair correlation function
on $\R^d\times \R^d$. Assume that $g_n(u,v)=g(a_n^{1/d} u, a_n^{1/d} v)$ or
that $g_n(u,v)-1=(g(u,v)-1)/a_n$, then \eqref{hyp gn} is
satisfied whenever $\sup_v \int_{\R^d} |g(u,v)-1|\der u<\infty$.
\end{lem}
\begin{proof}
In the first case
\begin{align*}
a_n\int_{A^2} |g_n(u,v)-1|\der u\der v 
&= 
\int_A \int_{a_n^{1/d} A} |g(u,a_n^{1/d} v)-1|\der u\der v \\
&\leq   \int_A \int_{\R^d} |g(u,a_n^{1/d} v)-1|\der u\der v \\
&\leq |A| \sup_v \int_{\R^d} |g(u,v)-1|\der u.
\end{align*}
In the second case
\begin{align*}
a_n\int_{A^2} |g_n(u,v)-1|\der u\der v 
&=  
\int_{A^2} |g(u,v)-1|\der u\der v \leq |A| \sup_v \int_{\R^d} |g(u,v)-1|\der u.
\end{align*}

\end{proof}

\medskip

\noindent {\it Example 1}: If $X_n$ is an inhomogeneous Poisson
point process with intensity $\lambda_n=a_n \lambda$, then
$g_n(u,v)=1$ and \eqref{hyp gn} is obviously satisfied. 

 \medskip

\noindent {\it Example 2}:  Let $X_n$ be a Neyman-Scott process,
defined by $X_n=\bigcup_{c\in C_n} Y_c$ where $C_n$ is a
homogeneous Poisson point process (of cluster centres) and given
$C_n$,  $Y_c$ are independent inhomogeneous Poisson point
processes (of offsprings' clusters). In a first scenario, similar
to Example 1 in \cite{Choiruddin21}, assume that the intensity of
$C_n$ is $a_n$ and that the intensity of $Y_c$ is
$k(.-c)\lambda(.)$, where $k$ is a symmetric density on $\R^d$.
Then it is easily derived that $\lambda_n=a_n\lambda$ and
$g_n(u,v)=1+k\star k(v-u)/a_n$, where $\star$ denotes
convolution. In this first scenario, there are more and more
clusters as $a_n\to\infty$, but each of them keeps the same
characteristics in terms of mean number of offsprings and spread.
Note that in this case $g_n(u,v)-1=(g(u,v)-1)/a_n$ where
$g(u,v)=1+k\star k(v-u)$. In a second scenario, assume that the
intensity of $C_n$ is $a_n$  and that the intensity of $Y_c$ is
$a_n k(a_n^{1/d}(.-c))\lambda(.)$. Here, as $a_n\to\infty$, there are
more and more clusters and each cluster is smaller and smaller.
We have  in this case $\lambda_n=a_n\lambda$ and
$g_n(u,v)=g(a_n^{1/d} u, a_n^{1/d} v)$. In both scenarios, Lemma~\ref{lem gn}
applies, where the main condition holds true whenever $k$ is
compactly supported or fast decaying.

 \medskip

\noindent {\it Example 3}: Let $X_n$ be a  LGCP (log Gaussian Cox
process, see, e.g., \cite{baddeley15book}) driven by a Gaussian random field with mean
$\mu_n(.)=\log(a_n)+\mu(.)$, for some function $\mu$, and with
covariance function $c_n(u,v)=c(a_n^{1/d} (v-u))$ where $c$ is a
positive definite function.  Then $\lambda_n=a_n\lambda$
where $\lambda(.)=\exp(\mu(.)+c(0)/2)$ and $g_n(u,v)=g(a_n^{1/d} u,a_n^{1/d}
v)$ where $g(u,v)=\exp(c(v-u))$. Lemma~\ref{lem gn} applies and
the condition therein is satisfied if $c(u)\to 0$ as
$|u|\to\infty$ and $\int |c(u)|\der u<\infty$.

 \medskip

\noindent {\it Example 4}:  Let $X$ be a hardcore point process
with intensity $\lambda(.)$ and hardcore radius $R>0$, assuming
its existence. $X$ can for instance correspond to  an
inhomogeneous Mat\'ern hardcore model of type-I or type-II,
 see, e.g., \cite{baddeley15book}. 
  Let $X_n$ be the superposition of $a_n$ i.i.d. point processes following the same distribution as $X$ (assuming that $a_n\in \N$). This process has intensity
  $\lambda_n=a_n\lambda$  and
  pcf $g_n(u,v)-1=(g(u,v)-1)/a_n$, where $g$ is the pcf of $X$.
  Lemma~\ref{lem gn} applies and the condition on $g$ is for
  instance satisfied for the aforementioned hardcore Mat\'ern
  models.

 \medskip

\noindent {\it Example 5}: In the same spirit as in the previous
example, consider a DPP $X$ on $\R^d$ with kernel
$K(u,v)=\sqrt{\lambda(u)\lambda(v)}K_0(u,v)$ where $K_0(u,u)=1$,
assuming its existence (see \cite{Lavancier}). Then the superposition of $a_n$ i.i.d. processes distributed as $X$ has intensity  $\lambda_n=a_n\lambda$ and 
pcf $g_n(u,v)-1=(g(u,v)-1)/a_n$, where $g(u,v)=1-|K_0(u,v)|^2$ is the
pcf of $X$.  Lemma~\ref{lem gn} applies  again and the condition
on $g$ is satisfied if  $\sup_v \int_{\R^d}|K_0(u,v)|^2
\der u<\infty$, which holds true for most standard DPP kernels used
in spatial statistics.

\section{Simulation details}\label{appendix_simus}
We provide in this appendix some details on the simulation study of Section~\ref{xgboost}.
Following \cite{Lu06082025}, we consider ten covariates $z_1,\dots,z_{10}$ on the unit square generated as independent copies  of a centred Gaussian random
field with variance $1$ and exponential covariance with scale $0.1$. We then define the two following intensities:
$$\lambda(x) = c \exp\left(\beta \left(z_1(x) + z_2(x)\right) \right),$$
$$\lambda'(x) = c' \exp\left(\beta \left(z_1(x) + 
            \frac{z_2(x)z_3(x)}{2} + 
            \frac{\exp(z_4(x))}{6} + 
            \frac{z_5(x)}{2} + 
            3 \sin(z_6(x))\right) \right),$$
where the constants $c$ and $c'$ are such that the point processes with intensities $\lambda$ and $\lambda'$ generate 400 points on average, i.e. $\int \lambda=\int \lambda'=400$, and where $\beta$ controls the extent of spatial variation. 

For all considered point process models, we adopt three scenarios of estimation: (1) the intensity is $\lambda$ and the estimator uses $z_1$ and $z_2$ as covariates, so that there is no nuisance covariate; (2) the intensity is $\lambda'$ and the estimator uses all ten available covariates $z_1,\dots,z_{10}$, so that four nuisance covariates are present; and (3) the intensity is $\lambda$ and the estimator uses all ten available covariates $z_1,\dots,z_{10}$,  including 8 nuisance covariates.

In each scenario, four models are considered. The first model is a Poisson point process with intensity $\lambda$ or $\lambda'$ depending on the scenario. The second one is a log-Gaussian Cox process driven  by a Gaussian random field with exponential kernel, variance $\sigma^2$, scale parameter $h/100$, and
with mean chosen so that the process matches the prescribed intensity, $\lambda$ or $\lambda'$. The third model is a Neyman-Scott point process obtained as the thining of a Thomas point process, where the parents form a homogeneous Poisson point process with intensity   $\kappa$,  where the number of offsprings per parent follows a Poisson distribution with mean $\max_x(\lambda(x))$ (or $\max_x(\lambda'(x))$) and their distribution relative to their parent is an isotropic bivariate Gaussian distribution with standard deviation $h/100$. The thining probability is such that the resulting intensity is $\lambda$ or $\lambda'$ depending on the scenario. Finally, the fourth model is a determinantal point process with kernel $K(x,y)=\sqrt{\lambda(x)}K_0(y-x)\sqrt{\lambda(y)}$ where $K_0(u)=\exp(-\|u|^2/\alpha^2)$ and $\alpha=1/\sqrt{\pi\max_x(\lambda(x))}$, and similarly with $\lambda'$ instead of $\lambda$ for scenario 2.

\end{appendix}

\begin{acks}[Acknowledgments]
The authors thank Nicolas Chenavier from Universit\'e du Littoral
C\^ote d'Opale for fruitful discussions and comments on random
tessellations. They also thank the three anonymous referees for helping to improve an earlier version of this article.
\end{acks}

\bibliographystyle{imsart-number}
\bibliography{bib_RF}

\end{document}